%% file: paper.tex
\documentclass{article}

\usepackage{microtype}
\usepackage{graphicx}
\usepackage{subfigure}
\usepackage{booktabs} 
\usepackage{hyperref}

\usepackage[accepted]{icml2023}

\usepackage{amsmath}
\usepackage{amssymb}
\usepackage{mathtools}
\usepackage{amsthm}
\usepackage{physics} 

\theoremstyle{plain}
\newtheorem{theorem}{Theorem}[section]
\newtheorem{proposition}[theorem]{Proposition}
\newtheorem{lemma}[theorem]{Lemma}
\newtheorem{corollary}[theorem]{Corollary}
\theoremstyle{definition}

\theoremstyle{remark}

\input{macros.tex}

\icmltitlerunning{
	Gibbsian Polar Slice Sampling
}

\begin{document}
	
\twocolumn[
\icmltitle{
	Gibbsian Polar Slice Sampling
}




\begin{icmlauthorlist}
	\icmlauthor{Philip Schär}{jena}
	\icmlauthor{Michael Habeck}{jena}
	\icmlauthor{Daniel Rudolf}{passau}
\end{icmlauthorlist}

\icmlaffiliation{jena}{Microscopic Image Analysis Group, Friedrich Schiller University Jena, Jena, Germany}
\icmlaffiliation{passau}{Faculty of Computer Science and Mathematics, University of Passau, Passau, Germany}

\icmlcorrespondingauthor{Daniel Rudolf}{daniel.rudolf@uni-passau.de}

\icmlkeywords{MCMC, slice sampling, polar slice sampling, heavy-tailed distributions, funneling}

\vskip 0.3in
]

\printAffiliationsAndNotice{} 

\begin{abstract}
	Polar slice sampling \cite{PolarSS} is a Markov chain approach for approximate sampling of distributions that is difficult, if not impossible, to implement efficiently, but behaves provably well with respect to the dimension.
	By updating the directional and radial components of chain iterates separately, we obtain a family of samplers that mimic polar slice sampling, and yet can be implemented efficiently.
	Numerical experiments in a variety of settings indicate that our proposed algorithm outperforms the two most closely related approaches, elliptical slice sampling \cite{EllipticalSS} and hit-and-run uniform slice sampling \cite{MacKayBook}. We prove the well-definedness and convergence of our methods under suitable assumptions on the target distribution.
\end{abstract}

\section{Introduction} \label{Sec:Intro}

Bayesian inference heavily relies on efficient sampling schemes of posterior distributions that are defined on high-dimensional spaces with probability density functions that can only be evaluated up to their normalizing constants. We develop a Markov chain method that can be used to approximately sample a large variety of target distributions. 
For convenience we frame the problem in a black-box setting. We assume that the distribution of interest $\nu$ is defined on $(\R^d,\B(\R^d))$ and given by a density, i.e.~a measurable function $\varrho_{\nu}: \R^d \ra \R_+ := \coint{0}{\infty}$, such that
\begin{equation*}
	\nu(A)
	= \frac{\int_A \varrho_{\nu}(x) \d x}{\int_{\R^d} \varrho_{\nu}(x) \d x}, 
	\qquad A \in \B(\R^d).
\end{equation*}
In the following, knowledge of the normalization constant of $\varrho_{\nu}$, i.e.~the denominator of the above ratio, is not required. Since exact sampling is generally infeasible, we aim to produce approximate samples from $\nu$, i.e.~realizations of random variables whose distributions are, in some sense, close to $\nu$.

To pursue this goal, we rely on Markov chain Monte Carlo (MCMC), which implements an irreducible and aperiodic Markov kernel $P$ that leaves $\nu$ invariant. For such a kernel, well-established theory shows that the distribution of iterates $X_n$ of a Markov chain $(X_n)_{n \in \N_0}$ with transition kernel $P$ converges to $\nu$ as $n \ra \infty$. An MCMC method generates a realization $(x_n)_{1\leq n \leq N} \subset \R^d$ of $X_1,\dots,X_N$ from $(X_n)_{n\in\N_0}$ and uses some or all of the realized chain iterates $x_n$ as approximate samples of $\nu$. Here, we focus on \textit{slice sampling} MCMC methods, that use auxiliary ``slice'' or threshold random variables $(T_n)_{n\in\N_0}$. In general, $T_n$ given $X_{n-1}$ follows a uniform distribution over an interval that depends on $X_{n-1}$.
%
Specifically, we consider \textit{polar slice sampling} (PSS), which was proposed by \citet{PolarSS}. 
PSS factorizes the target density $\varrho_{\nu}$ into
\begin{align}
	\begin{split}
		&\varrho_{\nu}(x) = \varrho_{\nu}^{(0)}(x) \varrho_{\nu}^{(1)}(x), \\
		&\varrho_{\nu}^{(0)}(x) = \norm{x}^{1-d}, \\
		&\varrho_{\nu}^{(1)}(x) = \norm{x}^{d-1} \varrho_{\nu}(x)
	\end{split}
	\label{Eq:den_fac}
\end{align}
for $x \neq 0$, where $\norm{\cdot}$ is the Euclidean norm on $\R^d$. Given an initial value $x_0 \in \R^d$ with $x_0 \neq 0$ and $\varrho_{\nu}(x_0) > 0$, PSS recursively realizes the $n$-th chain iterate $x_n$ from the $(n-1)$-th iterate $x_{n-1}$ as follows: An auxiliary variable $t_n$ is chosen as a realization of\footnote{With $\U(G)$ we denote the uniform distribution on a set $G$ w.r.t.~some reference measure that is clear from the context, as it is always either the Lebesgue measure, the surface measure on the $(d-1)$-sphere or a product measure of the two.}
\begin{equation*}
	T_n \sim \U(\! \ooint{0}{\varrho_{\nu}^{(1)}(x_{n-1})}) .
\end{equation*}
Given $t_n$, the next chain iterate is generated using polar coordinates $X_n = R_n \Theta_n$ where the \textit{radius} $R_n$ is a random variable on $\R_+$ and the \textit{direction} $\Theta_n$ is a random variable on the $(d-1)$-sphere 
\begin{equation*}
	\sph^{d-1}:=\{ \theta\in \R^d \mid \norm{\theta} = 1\} \subset \R^d .
\end{equation*}
To conform with the general slice sampling principle of \citet{BesagGreen}, the variables $(R_n, \Theta_n)$ need to be sampled from the joint uniform distribution
\begin{equation}
	\U(\{(r,\theta) \in \R_+ \times \sph^{d-1} \mid \varrho_{\nu}^{(1)}(r \theta) > t_n\}) .
	\label{Eq:PSS_Xup}
\end{equation}
Standard slice sampling theory then guarantees that the resulting transition kernel has invariant distribution $\nu$. For the convenience of the reader, in Appendix~\ref{Sec:Deri} we elaborate on how PSS can be derived from the general slice sampling principle of \citet{BesagGreen}. In the following, we refer to the process of sampling $X_n$ given $T_n$ as \textit{$X$-update} of PSS.

The theoretical analysis of PSS by \citet{PolarSS} offered performance guarantees for approximate sampling that are dimension-independent for rotationally invariant, log-concave target densities. Overall, their analysis suggests that PSS works generally robustly in high-dimensional scenarios. Despite this, PSS received little attention in the MCMC literature during the twenty years since its publication. We believe that this lack of engagement is not the result of PSS performing poorly on paper, but rather that of practical challenges in efficiently implementing it. Concurrent work by \citet{PSS_paper} supports this view. They prove -- again in the rotationally invariant, log-concave setting -- dimension-independent spectral gap estimates for PSS, which imply dimension-independence w.r.t.~the mean squared error and the asymptotic variance within the central limit theorem of the MCMC average of a summary function.

The practical challenge in implementing PSS is that the polar variables $(R_n, \Theta_n)$ need to be jointly drawn uniformly from a high-dimensional set that often has a complicated structure. Therefore, this step is usually implemented by an acceptance/rejection scheme using uniform samples from a tractable superset. In moderate to high dimensions, for target densities $\varrho_{\nu}$ that are not rotationally invariant, the fraction of directions $\theta \in \sph^{d-1}$ for which the set
\begin{equation*}
	\{r \in \R_+ \mid \varrho_{\nu}^{(1)}(r \theta) > t\}
\end{equation*}
is non-empty becomes tiny for most thresholds $t$ occurring as realizations of $T_n$. This usually leads to an impractically low acceptance rate of the aforementioned acceptance/rejection scheme, such that -- in expectation -- an astronomically large number of proposals needs to be drawn during a single transition. In other words, although a valid implementation of PSS is available in principle, the iterations of the sampler are computationally inefficient resulting in exceedingly long simulation runs.

To address this deficiency, we develop an MCMC framework that imitates PSS, but is guaranteed to run in a computationally efficient manner. Imitating here refers to keeping the PSS structure and splitting the difficult joint uniform sampling of radius and direction of \eqref{Eq:PSS_Xup} into separate steps. Intuition suggests that if the resulting transition mechanism is close to the original version of PSS, then also some of its desirable convergence properties will be inherited. The eventually proposed MCMC algorithm is essentially tuning-free and explores the state space remarkably quickly. We provide a basic theoretical underpinning of our method, proving that it asymptotically samples from the target distribution $\nu$ under mild regularity conditions. Moreover, we illustrate its potential to improve upon related methods through a series of numerical experiments.

The remainder of this paper is structured as follows: In Section~\ref{Sec:GPSS} we propose our modifications of PSS that we term \textit{Gibbsian polar slice sampling}. To provide a better understanding, we present different variants that culminate in a Gibbsian polar slice sampler that incorporates a stepping-out and shrinkage procedure. In Section \ref{Sec:Theory} we provide theoretical support for our methods. We discuss possible extensions and alternative mechanisms that can also be used in our framework in Section~\ref{Sec:AltMech} and comment on related approaches in Section \ref{Sec:RelatedWork}. In Section \ref{Sec:Experiments} we present a number of numerical experiments comparing our algorithm to the two most closely related ones that are similarly feasible to use in practice. We conclude with a short discussion in Section \ref{Sec:Discussion}.

\section{Gibbsian Polar Slice Sampling} \label{Sec:GPSS}

The idea is to decompose the $X$-update of PSS by replacing the joint sampling of radius $r_n$ and direction $\theta_n$ with separate updates of both variables in a Gibbsian fashion. 

\subsection*{Variant~1}
Our initial modification of PSS is mostly of theoretical interest. First, we only update the directional component of the last chain iterate $x_{n-1}$, resulting in an intermediate state $r_{n-1} \theta_n$, where $r_{n-1} = \norm{x_{n-1}}$ and $\theta_n$ is a realization of
\begin{equation}
	\Theta_n \sim \U(\{\theta \in \sph^{d-1} \mid \varrho_{\nu}^{(1)}(r_{n-1} \theta) >t_n \}) .
	\label{Eq:GPSS_dir}
\end{equation}
We then update the radial component of the intermediate state, resulting in the new chain iterate $x_n := r_n \theta_n$, where $r_n$ is a realization of
\begin{equation}
	R_n \sim \U(\{r \in \R_+ \mid \varrho_{\nu}^{(1)}(r \theta_n) > t_n\}) .
	\label{Eq:GPSS_rad}
\end{equation}
In contrast to standard Gibbs sampling which cycles over coordinate-wise updates of the current state, we rely on a systematic conditional renewal in terms of the polar transformation components given as radius and direction. Although the modification does not solve the runtime issue, it lays the groundwork for further algorithmic improvements. 

At this stage, we would already like to mention Theorem~\ref{Thm:GPSS_inv} which states that $\nu$ is the invariant distribution of
the transition kernel of the 1st variant of \textit{Gibbsian polar slice sampling} (GPSS). Therefore, GPSS provides a correct MCMC method for targeting $\nu$.

We note that by randomizing the deterministic updating scheme (i.e. rather than always updating the direction first and then the radius, both are sampled in random order), one can obtain the stronger statement that the transition kernel is reversible w.r.t.~$\nu$.

\subsection*{Variant~2}
The direction update is still a challenging implementation issue, since it requires sampling of a uniform distribution over a $(d-1)$-dimensional set with $d$ possibly being large. Therefore, to sample the next direction, we suitably use (ideal) spherical slice sampling \cite{SphericalSS}, a recently developed MCMC method for (approximate) sampling from distributions on $\sph^{d-1}$. The radius update is not changed in this variant (since it consists of a $1$-dimensional uniform distribution sampling step). 

Let $\sph_{\theta}^{d-2}$ be the \textit{great subsphere w.r.t.~$\theta$}, i.e.~the set
\begin{equation*}
	\{ \vartheta \in \sph^{d-1} \mid \theta^T \vartheta = 0 \}
\end{equation*}
of directions in $\sph^{d-1}$ that are orthogonal to $\theta$. In \cite{SphericalSS} it is shown that one can sample uniformly from $\sph_{\theta}^{d-2}$ as follows: Draw $V_1 \sim \Nc_d(0,I_d)$, set $V_2 := V_1 - (\theta^T V_1) \theta$ and finally $V_3 := V_2 / \norm{V_2}$, then $V_3 \sim \U(\sph_{\theta}^{d-2})$. Furthermore, for any $y \in \sph_{\theta}^{d-2}$ the set
\begin{equation*}
	\{ \theta \cos(\omega) + y \sin(\omega) \mid \omega \in \coint{0}{2\pi} \}
\end{equation*}
is the unique great circle in $\sph^{d-1}$ that contains both $\theta$ and $y$.

A single iteration of this variant imitates the $X$-update of PSS as follows: First, a reference point $y$ is drawn uniformly from the great subsphere w.r.t.~the direction $\theta_{n-1}$ of the previous sample. The new direction $\theta_n$ is then sampled uniformly from the great circle of $\sph^{d-1}$ running through both $\theta_{n-1}$ and $y$, intersected with 
\begin{equation*}
	\{\theta \in \sph^{d-1}\mid \varrho_{\nu}^{(1)}(r_{n-1} \theta)>t_n\}.
\end{equation*}
Then, a new radius is chosen by sampling \eqref{Eq:GPSS_rad}.

\subsection*{Variant 3: The Concrete Algorithm}

In practice, the univariate direction and radius updates of the 2nd variant of GPSS still need to be implemented as acceptance/rejection schemes that might exhibit low acceptance rates. Therefore, the final variant of GPSS replaces both updates with adaptive procedures that are essentially guaranteed to be fast and result in an algorithm that empirically converges against the correct target distribution. 

In the radius update, we use the stepping-out and shrinkage procedure as proposed for uniform slice sampling by \citet{SSNeal}. In the direction update, we incorporate a shrinkage procedure. Actually, our direction update can be interpreted as running the shrinkage-based spherical slice sampler \cite{SphericalSS}. This ultimate variant is readily implemented by Algorithms~\ref{Alg:GPSS}, \ref{Alg:OS} and \ref{Alg:RS}.

Briefly a single iteration works as follows: An element $y$ of the great subsphere is determined as in variant 2 and then the new direction $\theta_n$ is sampled via a shrinkage procedure on the great circle of $\sph^{d-1}$ running through both $\theta_{n-1}$ and $y$. Finally, a new radius is chosen via a stepping-out and shrinkage procedure on the ray emanating from the origin in direction $\theta_n$, where shrinkage can be performed around the old radius $r_{n-1}$ because it satisfies the target condition by construction.

\begin{algorithm}[tb]
	\caption{Gibbsian Polar Slice Sampling}
	\label{Alg:GPSS}
	\begin{algorithmic}[1]
		\STATE {\bfseries Input:} target density $\varrho_{\nu}$, initial value $x_0 \in \R^d$ with $x_0 \neq 0$ and $\varrho_{\nu}(x_0) > 0$, initial interval length $w > 0$
		\STATE Define $\varrho_{\nu}^{(1)}: x \mapsto \norm{x}^{d-1} \varrho_{\nu}(x)$
		\STATE Set $r_0 := \norm{x_0}$ and $\theta_0 := x_0 / r_0$
		\FOR{$n=1,2,\ldots$}
		\STATE Draw $T_n \sim \U(\!\ooint{0}{\varrho_{\nu}^{(1)}(x_{n-1})})$, call result $t_n$
		\STATE $\theta_n := \text{Geodesic\_Shrinkage}(\varrho_{\nu}^{(1)}, t_n, r_{n-1}, \theta_{n-1})$
		\STATE $r_n := \text{Radius\_Shrinkage}(\varrho_{\nu}^{(1)}, t_n, r_{n-1}, \theta_n, w)$
		\STATE $x_n := r_n \theta_n$
		\ENDFOR
		\STATE {\bfseries return} $(x_n)_{n \geq 0}$
	\end{algorithmic}
\end{algorithm}

\begin{algorithm}[tb]
	\caption{Geodesic Shrinkage}
	\label{Alg:OS}
	\begin{algorithmic}[1]
		\STATE {\bfseries Input:} transform $\varrho_{\nu}^{(1)}$ of target density, current threshold $t_n$, current radius $r_{n-1}$, current direction $\theta_{n-1}$
		\STATE Draw $Y \sim \U(\sph^{d-1}_{\theta_{n-1}})$, call the result $y$
		\STATE Draw $\omega_{\max} \sim \U([0,2\pi])$, set $\omega_{\min} := \omega_{\max} - 2 \pi$
		\REPEAT
		\STATE Draw $\Omega \sim \U(\!\ccint{\omega_{\min}}{\omega_{\max}})$, call result $\omega$
		\STATE Set $\theta_n := \theta_{n-1} \cos \omega + y \sin \omega$
		\STATE {\bfseries if} $\omega < 0$ {\bfseries then} $\omega_{\min} := \omega$ {\bfseries else} $\omega_{\max} := \omega$
		\UNTIL{$\varrho_{\nu}^{(1)}(r_{n-1} \theta_n) > t_n$}
		\STATE {\bfseries return} $\theta_n$
	\end{algorithmic}
\end{algorithm}

\begin{algorithm}[tb]
	\caption{Radius Shrinkage}
	\label{Alg:RS}
	\begin{algorithmic}[1]
		\STATE {\bfseries Input:} transform $\varrho_{\nu}^{(1)}$ of target density, current threshold $t_n$, current radius $r_{n-1}$, current direction $\theta_n$, initial interval length $w > 0$
		\STATE Sample $U \sim \U([0,1])$, call the result $u$
		\STATE $r_{\min} := \max(r_{n-1} - u \cdot w, 0)$
		\STATE $r_{\max} := r_{n-1} + (1-u) \cdot w$
		\WHILE{$r_{\min} > 0$ {\bfseries and} $\varrho_{\nu}^{(1)}(r_{\min} \theta_n) > t_n$}
		\STATE $r_{\min} := \max(r_{\min} - w, 0)$
		\ENDWHILE 
		\WHILE{$\varrho_{\nu}^{(1)}(r_{\max} \theta_n) > t_n$}
		\STATE $r_{\max} := r_{\max} + w$
		\ENDWHILE 
		\STATE Sample $R_n \sim \U([r_{\min},r_{\max}])$, call the result $r_n$ 
		\WHILE{$\varrho_{\nu}^{(1)}(r_n \theta_n) \leq t_n$}
		\STATE {\bfseries if} $r_n < r_{n-1}$ {\bfseries then} $r_{\min} := r_n$ {\bfseries else} $r_{\max} := r_n$
		\STATE Sample $R_n \sim \U([r_{\min},r_{\max}])$, call the result $r_n$
		\ENDWHILE
		\STATE {\bfseries return} $r_n$
	\end{algorithmic}
\end{algorithm}

\section{Validation -- Theoretical Support} \label{Sec:Theory}

We provide some basic theoretical underpinning of the proposed methods with a focus on the 2nd variant of GPSS. The reasons for this are threefold:
First, in principle this variant can also be implemented by using univariate acceptance/rejection schemes\footnote{For the direction update this is immediately possible, since $[0,2\pi]$ is a superset of the corresponding acceptance region. 
	For the radius update this is more complicated and may require some structural knowledge about $\varrho_{\nu}$.}. Second, we want to avoid any deep discussion about the stepping-out procedure that is involved in Algorithm~\ref{Alg:RS}. Third, since Algorithms~\ref{Alg:OS} and \ref{Alg:RS} both contain a shrinkage procedure, no explicit representation of the corresponding transition kernels is available to our knowledge. The proofs of the following statements can be found in Appendix \ref{Sec:Proofs}.

\begin{theorem} \label{Thm:GPSS_inv}
	The transition kernels corresponding to the 1st and 2nd variant of GPSS admit $\nu$ as invariant distribution.
\end{theorem}

To verify that $\nu$ is an invariant distribution, we argue that the direction and radius update are both individually reversible w.r.t.~$\nu$. This could also serve as a strategy to prove the invariance of the 3rd variant of GPSS. For the direction update (Algorithm~\ref{Alg:OS}), this can be done by virtue of results of \cite{SphericalSS,RevESS}. For the radius update, the interplay of stepping-out and shrinkage makes it difficult to prove reversibility rigorously. However, intuitively the arguments of \citet{SSNeal} apply.

Under weak regularity assumptions on $\varrho_{\nu}$ we also get a convergence statement.

\begin{theorem} \label{Thm:GPSS_variant2_conv}
	Assume that the target density $\varrho_{\nu}$ is strictly positive, i.e.~$\varrho_{\nu} \colon \R^d \to \ooint{0}{\infty}$. Then, for $\nu$-almost every initial value $x_0$, the distribution of an iterate $X_n$ of a Markov chain $(X_n)_{n \in \N_0}$ with $X_0 = x_0$ and transition kernel corresponding to the 2nd variant of GPSS converges to $\nu$ in total variation as $n \to \infty$.
\end{theorem}

Under an additional restrictive structural requirement the convergence result also holds for the 3rd variant of GPSS. Namely, if we assume that the radius shrinkage of Algorithm~\ref{Alg:RS} with threshold $t_n$, current radius $r_{n-1}$ and direction $\theta_n$ realizes sampling w.r.t.
\begin{equation*}
	\U(\{r \in \R_+ \mid \varrho_{\nu}^{(1)}(r \theta_n) > t_n\}).
\end{equation*}
For example, this is true if $\varrho_{\nu}^{(1)}$ is unimodal along rays. This actually is a scenario where the direction update relies on the shrinkage-based Algorithm~\ref{Alg:OS}, but the radius update is exact as in the 2nd variant of GPSS. 


\section{Alternative Transition Mechanisms} \label{Sec:AltMech}

We emphasize that our 3rd variant of GPSS is just one of many possible ways to create a valid and efficiently implementable method that is based on the 2nd variant. For example, many of the ideas in Section 2.4 of \cite{EllipticalSS} for modifying the shrinkage procedure used by elliptical slice sampling easily transfer to the shrinkage procedure used in the direction update of GPSS.
For the radius update, one could drop stepping-out, i.e.~just place an interval of size $w > 0$ randomly around the current $r_n$ and then use the shrinkage procedure right away. With the same choice of the hyperparameter $w$ this reduces the number of target density evaluations per iteration at the cost of also reducing the speed at which the sampler explores the distribution of interest. 
One could also replace our radius update by the update mechanism of latent slice sampling \cite{LatentSS}. In principle, one could even use Metropolis-Hastings transitions for either one of the updates.
As long as each transition mechanism leaves the corresponding distribution in either \eqref{Eq:GPSS_dir} or \eqref{Eq:GPSS_rad} invariant, the resulting sampler should be well-behaved. In this sense, GPSS provides a flexible framework that leads to a variety of different samplers.

\section{Related Work} \label{Sec:RelatedWork}

The radius update we use for GPSS in the 3rd variant has previously been considered in Section 4.2.2 of \cite{ThompsonPhD}, where it was suggested to alternate it with a standard update on $\R^d$, i.e.~one that does not keep the radius fixed. Intuitively, our GPSS improves upon this approach by introducing a dedicated direction update, which, by not wasting effort on exploring the entire sample space, can more efficiently determine a (hopefully) good new direction.

Furthermore, although GPSS was developed as imitating classical polar slice sampling, it also bears some resemblance to two other slice sampling methods, namely \textit{elliptical slice sampling} (ESS) \cite{EllipticalSS} and \textit{hit-and-run uniform slice sampling}\footnote{The idea of HRUSS is already formulated in the paragraph `How slice sampling is used in real problems' in Section 29.7 in \cite{MacKayBook}} (HRUSS). Both methods have been analyzed theoretically to some extent, ESS in \cite{EllSSConv,RevESS}, HRUSS in \cite{Latuszynski} and an idealized version of HRUSS in \cite{Ullrich}. Moreover, a sophisticated sampling algorithm that employs a large number of ESS-based Markov chains running in parallel has been suggested in \cite{GenEllSS}.

Both ESS and HRUSS follow the same basic principle as all other slice sampling approaches: In iteration $n$, they draw a threshold $t_n$ w.r.t.~the value of some fixed function $\varrho_{\nu}^{(1)}$ at the latest sample $x_{n-1}$. They then determine the next sample $x_n$ by approximately sampling from some distribution restricted to the \textit{slice} (or level set) of $\varrho_{\nu}^{(1)}$ at threshold $t_n$. 

For a given threshold, ESS draws an auxiliary variable from a mean zero multivariate Gaussian and performs shrinkage on the zero-centered ellipse running through both the auxiliary point and the latest sample, using the latter as the reference point for shrinkage. This is very similar to the direction update in the 3rd variant of GPSS\footnote{In fact, both approaches can even be implemented to use the same random variables, $d$ samples from $\Nc(0,1)$, for determining the one-dimensional object to perform shrinkage on.}, where we draw the auxiliary variable $y$ from the uniform distribution on the great subsphere w.r.t.~the direction $\theta_{n-1}$ of the latest sample $x_{n-1}$ and then perform shrinkage on the unique great circle of $\sph^{d-1}$ that contains both $y$ and $\theta_{n-1}$.

HRUSS on the other hand uses neither ellipses nor circles. For a given threshold, it proceeds by choosing a random direction (uniformly from $\sph^{d-1}$) and determining the next sample $x_n$ by performing stepping-out and shrinkage procedures on the line through the latest sample $x_{n-1}$ in this direction. This is obviously very similar to the radius update used in the 3rd GPSS variant. The two major differences are that the latter does not use an auxiliary variable to determine the direction along which it performs the update, and that HRUSS performs an update on an entire line, whereas GPSS only considers a ray (by requiring radii to be non-negative).

Based on these comparisons, we may expect an iteration of our 3rd GPSS variant to be roughly as costly as one of ESS and one of HRUSS combined. Various experiments suggest this to be a good rule of thumb, though GPSS actually tends to be faster than ESS if the latter is not well-tuned (we explain what this means in Section \ref{Sec:Experiments}). Although HRUSS is clearly the fastest among the three, it usually takes large amounts of very small steps, so that it tends to lag behind the other two samplers when considering metrics like effective sample size per time unit.

\section{Experiments} \label{Sec:Experiments}

We illustrate the strengths of our 3rd variant of GPSS by a series of numerical experiments in which we compare its performance with those of ESS and HRUSS\footnote{All of our experiments were conducted on a workstation equipped with an AMD Ryzen 5 PRO 4650G CPU.}. Source code allowing the reproduction (in nature) of our experimental results is provided in a github repository\footnote{\href{https://github.com/microscopic-image-analysis/Gibbsian_Polar_Slice_Sampling}{https://github.com/microscopic-image-analysis/Gibbsian\_Polar\_Slice\_Sampling}}. Some remarks on the pitfalls in implementing our method are provided in Appendix \ref{Sec:ImpNotes}, additional sampling statistics in Appendix \ref{Sec:SamStats}, an additional experiment on Bayesian logistic regression in Appendix \ref{Sec:LogReg}, and further illustrative plots in Appendix \ref{Sec:ExpRes}.

Before we discuss the individual experiments, some explanation of our general approach to using and comparing these methods is in order. The positive hyperparameter $w$ in HRUSS determines the initial size of the search interval in the stepping-out procedure, playing the same role as the eponymous parameter of GPSS in Algorithm~\ref{Alg:RS}. Therefore, when we compare the two samplers in some fixed setting, we always use the same value of $w$ for both of them. We note, however, that neither parameter has much of an influence on the sampler's performance as long as they are not chosen orders of magnitude too small. Accordingly, we did not carefully tune them in any of the experiments.

As ESS is technically only intended for posterior densities with mean zero Gaussian priors, whenever we are in a different setting, we artificially introduce a mean zero Gaussian prior and use the target divided by the prior as likelihood. As we will see in the following, the performance of the method can be quite sensitive to the choice of the covariance matrix for this artificial prior, so we sometimes consider both a ``naive'' (or ``untuned'') and a ``sophisticated'' (or ``tuned'') choice in our comparisons. 

To make the sophisticated choices, we typically rely on the fact that the radii of samples from $\Nc_d(0, I_d)$ scatter roughly around $\sqrt{d}$, but also on an understanding of the desired range of radii gained either through theoretical analysis of the target density or by examining samples generated by GPSS. In one of our experiments, where the variables are highly correlated in the sense that their joint density -- our target density -- is very far from rotational invariance, we even use the empirical covariance of samples generated by GPSS as the covariance matrix for ESS.


\subsection{Multivariate Standard Cauchy Distribution} \label{SubSec:std_Cauchy}

First we considered the multivariate generalization of the pathologically heavy-tailed standard Cauchy distribution, i.e.~the distribution $\nu$ with $\varrho_{\nu}(x) = (1 + \norm{x}^2)^{-(d+1)/2}$. We chose the sample space dimension to be $d=100$, initialized all samplers with $x_0 := (1,\ldots,1)^T$ and ran each of them for $N = 10^6$ iterations. Since $\nu$ is rotationally invariant and the covariance of naive ESS was already on a reasonable scale, we refrained from using a tuned version of ESS here.

As the canonical test statistics mean and covariance are undefined for $\nu$, we instead measured the sample quality by letting the samplers estimate a probability w.r.t.~$\nu$. For this we chose the probability of the event that $\norm{Z} > b$ and simultaneously $Z_1 > 0$, where $Z = (Z_1,\dots,Z_d)^T \sim \nu$ and $b > 0$. Described in formulas, the samplers estimated
\begin{equation}
	p(b) := \nu(\{ z \in \R^d \mid \norm{z} > b, \; z_1 > 0 \}) ,
	\label{Eq:p(b)}
\end{equation}
where $z_1$ is the first entry of $z = (z_1,\dots,z_d)^T \in \R^d$.
Note that the condition $\norm{Z} > b$ measures how well the sample radii reflect those that would occur in exact sampling from $\nu$, whereas the condition $Z_1 > 0$ detects if the sample directions (i.e.~the samples divided by their radii), are biased towards either one of the two half-spaces separated by the hyperplane through all but the first coordinate axes. Hence both radii and directions need to be sampled well in order for the estimate of $p(b)$ to quickly approach the true value.

\begin{figure}[tb]
	\begin{center}
		\includegraphics[width=0.5\textwidth]{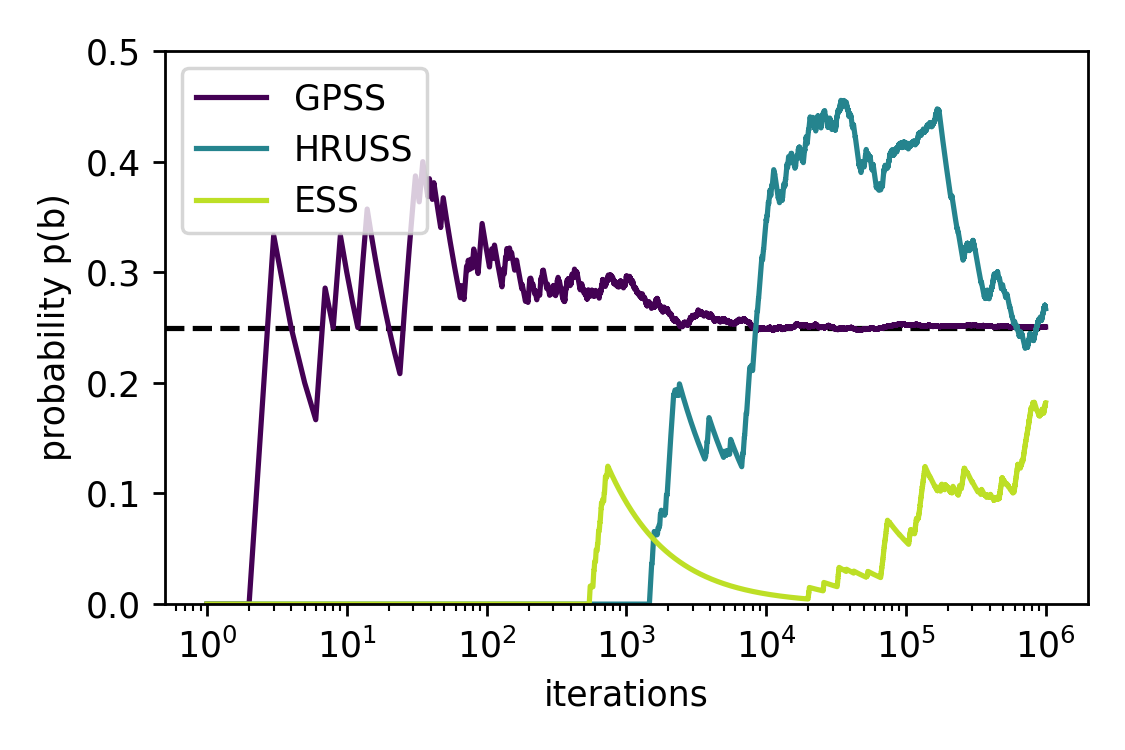}
	\end{center}
	\caption{Progression over iterations of the estimates of $p(b)$ from \eqref{Eq:p(b)} in the multivariate standard Cauchy experiment. Note the logarithmic scale. The dashed black line denotes the ground truth, which, using our knowledge of the target's rotational invariance, we could compute by numerically solving a one-dimensional integral. \label{Fig:std_Cauchy_conv}}
\end{figure}

The progressions of the samplers' estimates of $p(b)$ are shown in Figure \ref{Fig:std_Cauchy_conv}. It can be seen that those produced by GPSS converge to the true value orders of magnitude faster than those by both HRUSS and ESS. In Appendix \ref{Sec:ExpRes}, Figure \ref{Fig:std_Cauchy}, we provide a peek into the sampling behind these results by displaying the progression of radii and log radii over $N_{\text{window}} = 5 \cdot 10^4$ iterations. As exact sampling from $\nu$ is tractable (using samples from the multivariate normal and the $\chi^2$-distribution), we also display traces of exact i.i.d.~samples for comparison.

Figure \ref{Fig:std_Cauchy} suggests that the samples produced by GPSS are comparable in quality to i.i.d.~ones, and shows that those by HRUSS and ESS are certainly not. Upon closer examination, the reason for this becomes evident: Although all samplers take trips to the distribution's tails, those taken by HRUSS and ESS last several orders of magnitude longer than those of GPSS into the same distance. Consequently, HRUSS and ESS make their excursions to the extremely far-off parts of the tails (say, the points of radii in the thousands or more) with vanishingly small frequency. Thus GPSS needs a much shorter chain to properly reflect the tails of the target distribution. We acknowledge, however, that the advantage GPSS has over HRUSS in this setting would be considerably smaller if the target density was not centered around the origin.

\subsection{Hyperplane Disk} \label{SubSec:hyperplane}

Next we considered the target density
\begin{equation}
	\varrho_{\nu}(x)
	= \exp(- (\textstyle \sum_{i=1}^d x_i)^2 - \norm{x}^2)
	\label{Eq:hyperplane}
\end{equation}
for $x = (x_1,\ldots,x_d)^T \in \R^d$. Intuitively, the sum term within the density leads to a concentration of its probability mass around a hyperplane, i.e.~a $(d-1)$-dimensional subspace, given by the set of points $x \in \R^d$ for which the sum term vanishes. The norm term ensures that the function is integrable with Gaussian tails in all directions, which intuitively further concentrates the distribution around a circular disk within the hyperplane. We set $d = 200$ and initialized all chains in an area of high probability mass. We then ran each sampler for $N = 10^4$ iterations. To tune ESS, we took all $N$ samples generated by GPSS, computed their empirical covariance matrix and used it as the covariance parameter for the artificial Gaussian prior of ESS.

The results are shown in Appendix \ref{Sec:ExpRes}, Figure \ref{Fig:hyperplane}. The progression of the sample radii suggests that GPSS has a considerable advantage over all other approaches. However, impressions based on sample radii should be taken with some caution, since GPSS is the only method that specifically updates the radii during sampling. Accordingly, when considering empirical step sizes, the advantage of GPSS is less pronounced, but still present. For example, the mean step sizes are $\approx 5.0$ for GPSS, $\approx 3.5$ for tuned ESS, $\approx 2.4$ for untuned ESS and $\approx 0.6$ for HRUSS.

There are two drawbacks regarding the performance of tuned ESS, which by the aforementioned aspects is the closest competitor to our method in this setting. On the one hand, as noted before, we tuned ESS using the samples generated by GPSS, thus relying on the robust performance of GPSS to compute a good estimate of the target distribution's true covariance matrix. On the other hand, the computational overhead of sampling from a multivariate Gaussian with non-diagonal covariance matrix slows tuned ESS down significantly. As a result, tuned ESS consistently ran
slower than all the other samplers in this experiment.
The advantage of GPSS over tuned ESS is remarkable, not just because only the latter uses a proposal distribution adjusted to the shape of the target distribution, but also because the tails of the target are Gaussian, which in principle should benefit ESS (by virtue of its proposal distribution being Gaussian as well).

We attribute the relatively poor performance of HRUSS to the curse of dimensionality: As a result of the target density being narrowly concentrated around a hyperplane of a very high-dimensional space, from any given point in the target's high probability region there are very few directions in which one can take large steps without leaving the high-probability region. Nevertheless, HRUSS isotropically samples the direction along which it will move, such that most of the time only small steps along the sampled direction are allowed.

\subsection{Axial Modes} \label{SubSec:axial_modes}

Our third experiment was concerned with the target density
\begin{equation}
	\varrho_{\nu}(x)
	= \norm{x}_{\infty}^4 \exp(-\norm{x}_1) .
	\label{Eq:axial_modes_target}
\end{equation}
Due to the counteracting forces of $\infty$-norm and $1$-norm, $\varrho_{\nu}$ possesses $2d$ fairly isolated modes, $2$ along each coordinate axis, see Figure \ref{Fig:axial_modes_target} (appendix) for an illustration.

This particular structure enables an interesting quantitative diagnostic for the performance of MCMC methods targeting this distribution: For any given $x = (x_1,\ldots,x_d)^T \in \R^d$, one can use the absolute values of its components to assign it to a pair of modes (that lie on the same coordinate axis) via
\begin{equation*}
	\text{axis}(x)
	:= \argmax_{1 \leq i \leq d} \abs{x_i} .
\end{equation*}
For a finite chain of samples $(x^{(n)})_{n=1,\ldots,N} \subset \R^d$ that was generated as the output of some MCMC method, numerous quantitative diagnostics can then be applied to the values $(\text{axis}(x^{(n)}))_{n=1,\ldots,N}$. For example, one can say that the chain \textit{jumped between modes} in step $i$ if and only if
\begin{equation*}
	\text{axis}(x^{(i)}) \neq \text{axis}(x^{(i-1)}).
\end{equation*}
One can then count the \textit{total number of jumps} within the $N$ iterations (which provides information about how quickly the chain moved back and forth between the mode pairs) and compare these values between different chains. Alternatively, one could compute the \textit{mean dwelling time}, i.e.~the average number of iterations the chain spent at a mode pair until jumping to the next. This may be a more helpful quantity than the total number of jumps, because it is essentially independent of the number $N$ of iterations. It may also be worthwhile to consider the \textit{maximum dwelling time}, i.e.~the largest number of iterations the chain spent at a mode pair without leaving, as this is more suitable than mean dwelling time and total number of jumps for detecting whether a chain occasionally gets stuck at a mode pair for excessively many iterations.

We ran each sampler for $N = 10^5$ iterations in each of the dimensions $d=10,20,\ldots,100$. As initialization we used $x_0 := (5,1,\ldots,1)^T \in \R^d$. For the ESS covariance we considered both the usual naive choice $\Sigma = I_d$ and the carefully hand-tuned choice $\Sigma = (5+d/10)^2 / d \cdot I_d$.

\begin{figure}[tb]
	\begin{center}
		\includegraphics[width=0.5\textwidth]{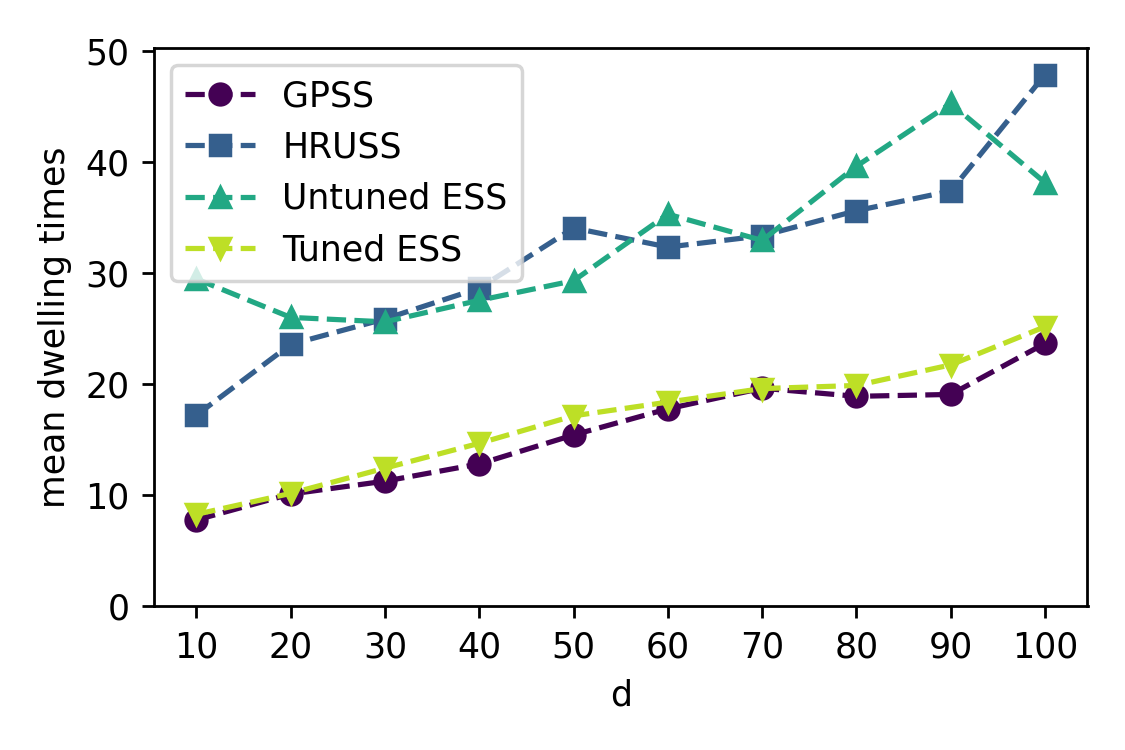}
	\end{center}
	\caption{Progression over dimensions $d$ of the mean dwelling time, determined based on $N=10^5$ iterations, in the axial modes experiment, as described in Section \ref{SubSec:axial_modes}. \label{Fig:axial_modes_jumps}}
\end{figure}

\begin{figure}[tb]
	\begin{center}
		\includegraphics[width=0.5\textwidth]{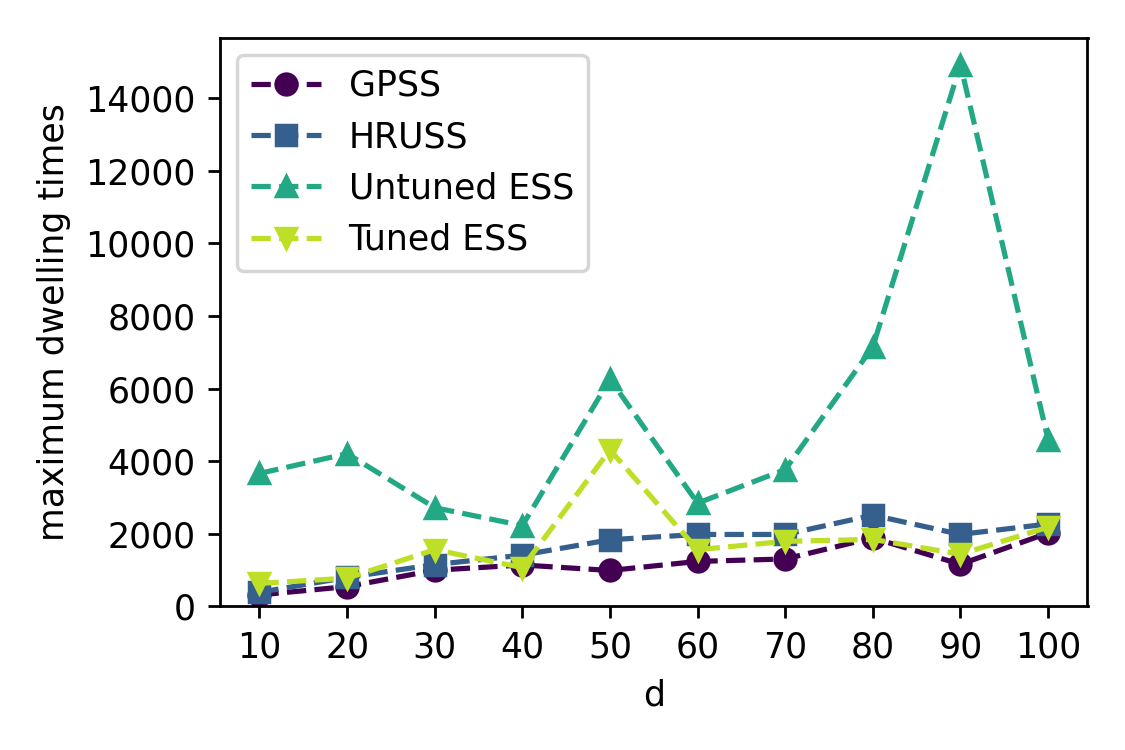}
	\end{center}
	\caption{Progression over dimensions $d$ of the maximum dwelling times within $N=10^5$ iterations in the axial modes experiment, as described in Section \ref{SubSec:axial_modes}. \label{Fig:axial_modes_max_dts}}
\end{figure}

In Figures \ref{Fig:axial_modes_jumps} and \ref{Fig:axial_modes_max_dts} we display for each sampler the progression over the dimensions $d$ of mean and maximum dwelling time. In terms of mean dwelling time, GPSS has a clear advantage over both untuned ESS and HRUSS, only the carefully tuned ESS is competitive with it. In terms of maximum dwelling time, GPSS is slightly ahead of all other samplers, even tuned ESS.
In Appendix \ref{Sec:ExpRes}, Figure \ref{Fig:axial_modes_fixed_dim} we provide a peek into the sampling behind these results by displaying the progression of currently visited mode pair in the last $N_{\text{window}} = 2 \cdot 10^4$ iterations of the samplers' runs for the highest dimension $d=100$.

\subsection{Neal's Funnel} \label{SubSec:funnel}

In our fourth experiment, we considered an arguably even more challenging target density that was originally proposed by \citet{SSNeal} and is commonly termed \textit{Neal's funnel}. It is given by
\begin{equation}
	\varrho_{\nu}(x)
	= \Nc(x_1; 0,9) \prod_{i=2}^d \Nc(x_i; 0, \exp(x_1)) ,
	\label{Eq:funnel}
\end{equation}
where we write $x = (x_1,\ldots,x_d)^T \in \R^d$ and denote by $\Nc(z; \mu,\sigma^2)$ the density of $\Nc(\mu,\sigma^2)$ evaluated at $z$. As the name suggests, the density is shaped like a funnel, having both a very narrow and a very wide region of high probability mass, which smoothly transition into one another. Besides being a challenging target, the funnel is of particular interest to us because its marginal distribution in the first coordinate is simply $\Nc(0,9)$ and a sampler needs to explore both the narrow and the wide part of the funnel equally well in order to properly approximate this marginal distribution via the marginals of its samples (i.e.~the set containing each of its samples truncated after the first component).

We ran the slice samplers for the funnel in dimension ${d=10}$ (as suggested by Neal) and initialized all of them with $x_0 := (2,0,\ldots,0)^T \in \R^d$. For ESS we used the relatively well-tuned covariance $\Sigma = \text{diag}(9,70,\ldots,70)$.
For this experiment we also extend our comparison to a sampler that is only related to GPSS by the fact that it is also an MCMC method. Namely, we compare GPSS with the No-U-Turn Sampler (NUTS) \cite{NUTS}, which is widely regarded to be the current state-of-the-art in MCMC sampling\footnote{Note, however, that NUTS needs the target density to be differentiable and requires oracle access to its gradient, neither of which is true for the slice sampling methods we consider.}. We used the implementation of NUTS provided by the probabilistic programming Python library PyMC \cite{PyMC}.

In this experiment we study the convergence to several target quantities. To enable a fair comparison of the samplers, we took differences in the speed of the iterations into account. This was achieved by allocating a fixed time budget of $1$~minute to each sampler and letting it run for sufficiently many iterations to deplete this time budget, while tracking how much time had elapsed after each completed iteration\footnote{Except for NUTS, where we only approximated this by assuming the runtime per iteration to be constant.}. We assessed the convergence ten times per second, resulting in $600$ measurement times in total. Finally, we determined for each measurement time $t$ and each sampler $s$ the exact number of iterations $n_{t,s} \in \N$ it had completed up to that time (using the aforementioned logs) and estimated the target quantities from only the $n_{t,s}$ samples produced in these iterations.
As target quantities we considered mean, standard deviation, $0.001$-quantile and $0.999$-quantile. To generate estimates of these quantities from marginal samples, we simply used their empirical versions.

\begin{figure}[tb]
	\begin{center}
		\includegraphics[width=0.5\textwidth]{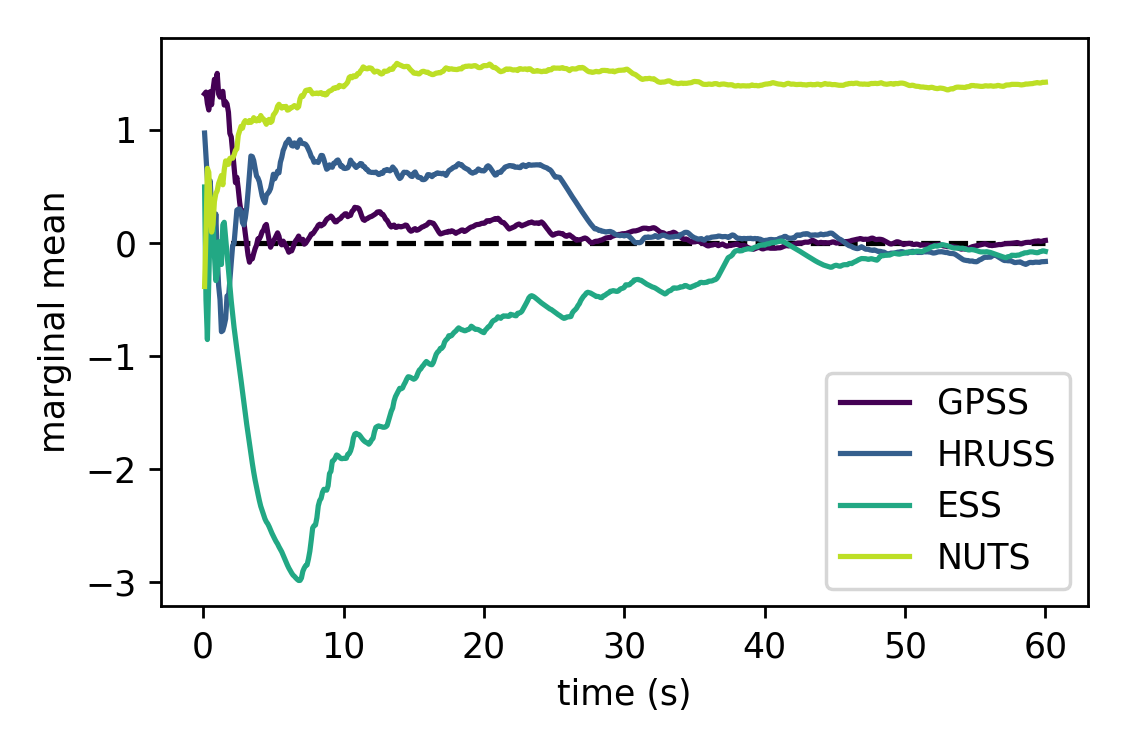}
	\end{center}
	\caption{Progression over runtime of the empirical mean of the marginal samples in the experiment on Neal's funnel \eqref{Eq:funnel}. The dashed black line denotes the ground truth. \label{Fig:funnel_mean}}
\end{figure}

\begin{figure}[tb]
	\begin{center}
		\includegraphics[width=0.5\textwidth]{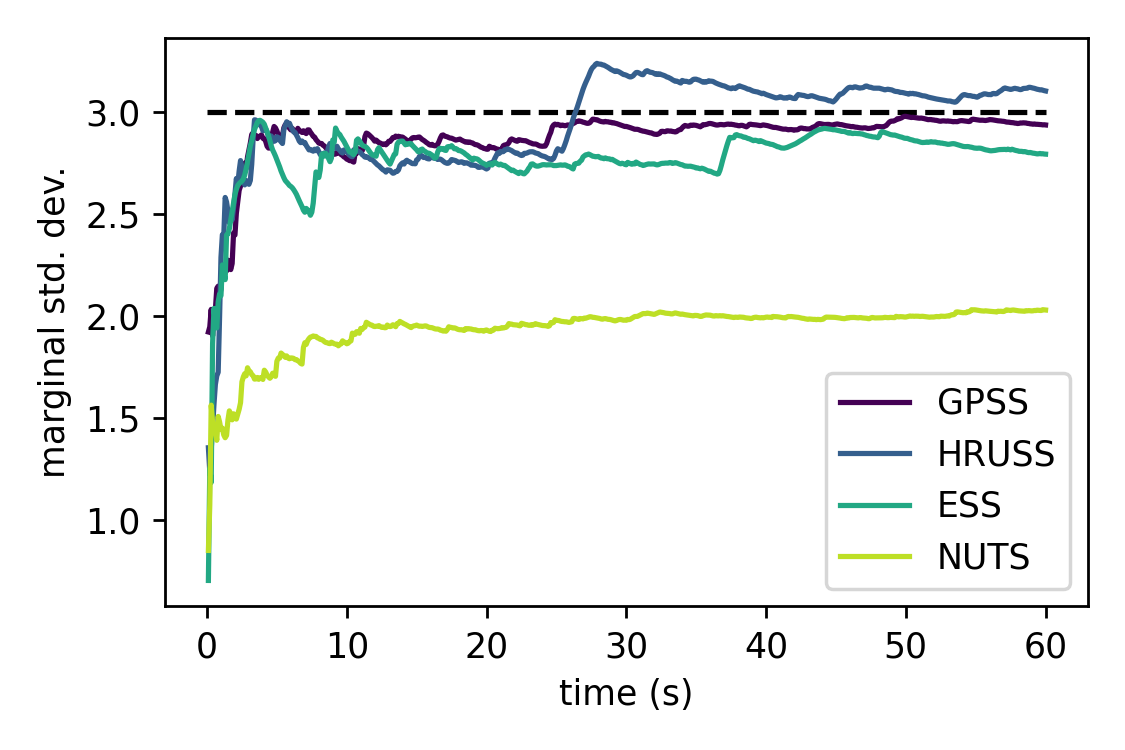}
	\end{center}
	\caption{Progression over runtime of the empirical standard deviation of the marginal samples in the experiment on Neal's funnel \eqref{Eq:funnel}. The dashed black line denotes the ground truth. \label{Fig:funnel_stddev}}
\end{figure}

\begin{figure}[tb]
	\begin{center}
		\includegraphics[width=0.5\textwidth]{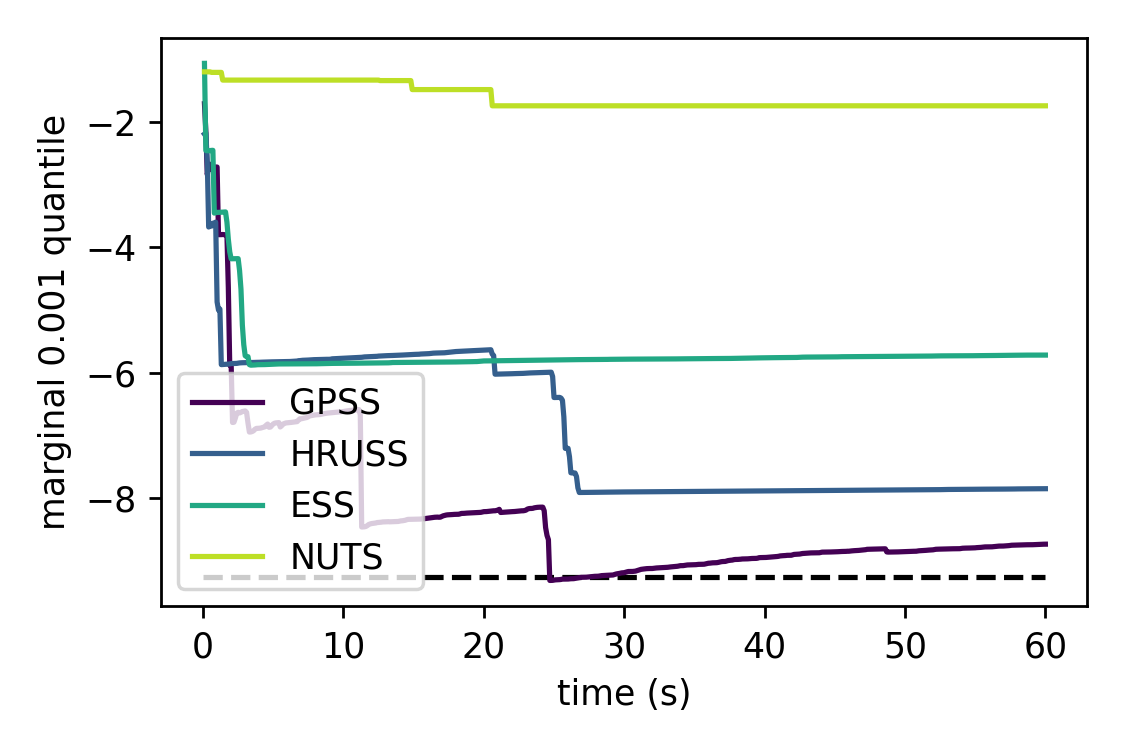}
	\end{center}
	\caption{Progression over runtime of the empirical $0.001$-quantile of the marginal samples in the experiment on Neal's funnel \eqref{Eq:funnel}. The dashed black line denotes the ground truth. \label{Fig:funnel_quant_q1}}
\end{figure}

\begin{figure}[tb]
	\begin{center}
		\includegraphics[width=0.5\textwidth]{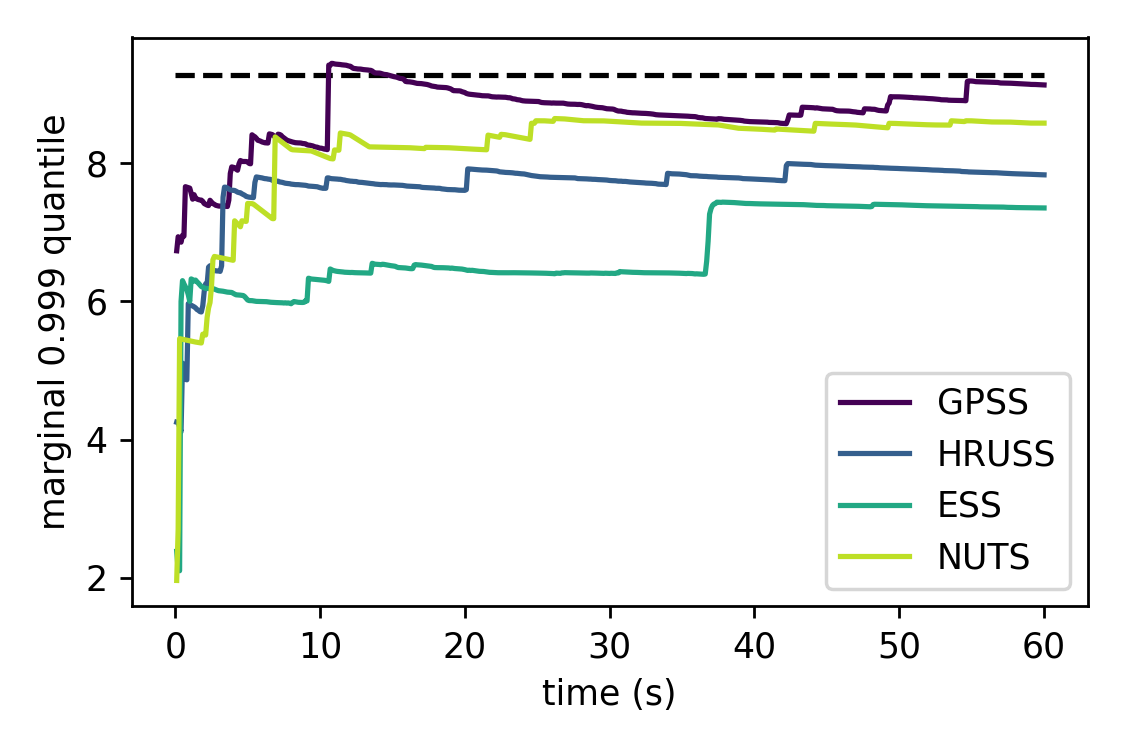}
	\end{center}
	\caption{Progression over runtime of the empirical $0.999$-quantile of the marginal samples in the experiment on Neal's funnel \eqref{Eq:funnel}. The dashed black line denotes the ground truth. \label{Fig:funnel_quant_q2}}
\end{figure}

The progression over time of each sampler's approximations to these target quantities are shown in Figures \ref{Fig:funnel_mean}, \ref{Fig:funnel_stddev}, \ref{Fig:funnel_quant_q1} and \ref{Fig:funnel_quant_q2}. Additionally, we provide the marginal histograms of all samples generated within the time budget as well as a peek into the progression of the marginal samples and sample radii in Appendix \ref{Sec:ExpRes}, Figure \ref{Fig:funnel_qualitative}.

It can be clearly seen from the first four figures that GPSS performs better overall than any of the other methods. We note that HRUSS is more competitive with it than ESS and that GPSS converges well despite completing only $3.39 \cdot 10^5$ iterations in the allotted time, which is less than half the $7.95 \cdot 10^5$ iterations completed by HRUSS, and about the same as the $2.94 \cdot 10^5$ completed by ESS. In other words, had we used the same number of iterations for all slice samplers -- like in other experiments -- the convergence results would attribute GPSS an even larger advantage.
Regarding the performance of NUTS, we observe that it is very successful at retrieving the 0.999-quantile, but, due to its refusal to enter the narrow part of the funnel, it performs worst among the four methods not just for the 0.001-quantile, but also for mean and standard deviation.

\section{Discussion} \label{Sec:Discussion}

We introduced a Gibbsian polar slice sampling (GPSS) framework as a general Markov chain approach for approximate sampling from distributions on $\R^d$ given by an unnormalized Lebesgue density. The efficiently implementable version that we propose is essentially tuning-free. It has only a single hyperparameter with little impact on the method's performance, as long as it is not chosen orders of magnitude too small. GPSS can quickly produce samples of high quality, and numerical experiments indicate advantages compared to related approaches in a variety of settings.

Although GPSS is generally quite robust, its performance slowly deteriorates when the distance between the target distribution's center of mass 
and the coordinate origin is increased. This could potentially be avoided by automatically centering the target on the origin. However, such a modification would likely follow an adaptive MCMC approach and therefore result in a method that no longer fits the ``strict'' MCMC framework.

Particularly good use cases for GPSS appear to be heavy-tailed target distributions. For example, one could apply GPSS to intractable posterior distributions resulting from Cauchy priors (perhaps on just some of the variables) and likelihoods that do not change the nature of the tails. As illustrated in Section \ref{SubSec:std_Cauchy}, GPSS can have enormous advantages over related methods when heavy tails are involved.
Another type of target for which GPSS could be of practical use are distributions with strong funneling, which naturally occur in Bayesian hierarchical models (cf.~\citet{SSNeal}). By nature, these targets call for methods with variable step sizes, because they contain both very narrow regions, which necessitate small step sizes, and very wide regions, in which much larger step sizes are advantageous to speed up the exploration of the sample space. Due to their use of variable step sizes, slice sampling methods might be considered a natural choice and, as demonstrated in Section \ref{SubSec:funnel}, GPSS appears to perform better than related slice samplers for such targets.

We envision that many more applications for GPSS will be found. Moreover, we think it may be possible to derive qualitative or even quantitative geometric convergence guarantees for GPSS. Aside from giving helpful insight into how well GPSS retains the desirable theoretical properties of PSS, this would further justify using GPSS in real-world applications, where the sample quality is often hard to validate.
Finally, we emphasize that the 2nd variant of GPSS, the intermediate step between idealized framework and efficiently implementable method, can also be used as the foundation for a variety of other hybrid samplers (cf.~Section \ref{Sec:AltMech}).

\section*{Acknowledgements}
We thank the anonymous referees for their suggestions. PS and MH gratefully acknowledge funding by the Carl Zeiss Foundation within the program ``CZS Stiftungsprofessuren'' and the project ``Interactive Inference''. We are grateful for the support of the DFG within project 432680300 -- SFB 1456 subprojects A05 and B02.

\appendix

\begin{figure}[tb]
	\begin{center}
		\includegraphics[width=0.4\textwidth]{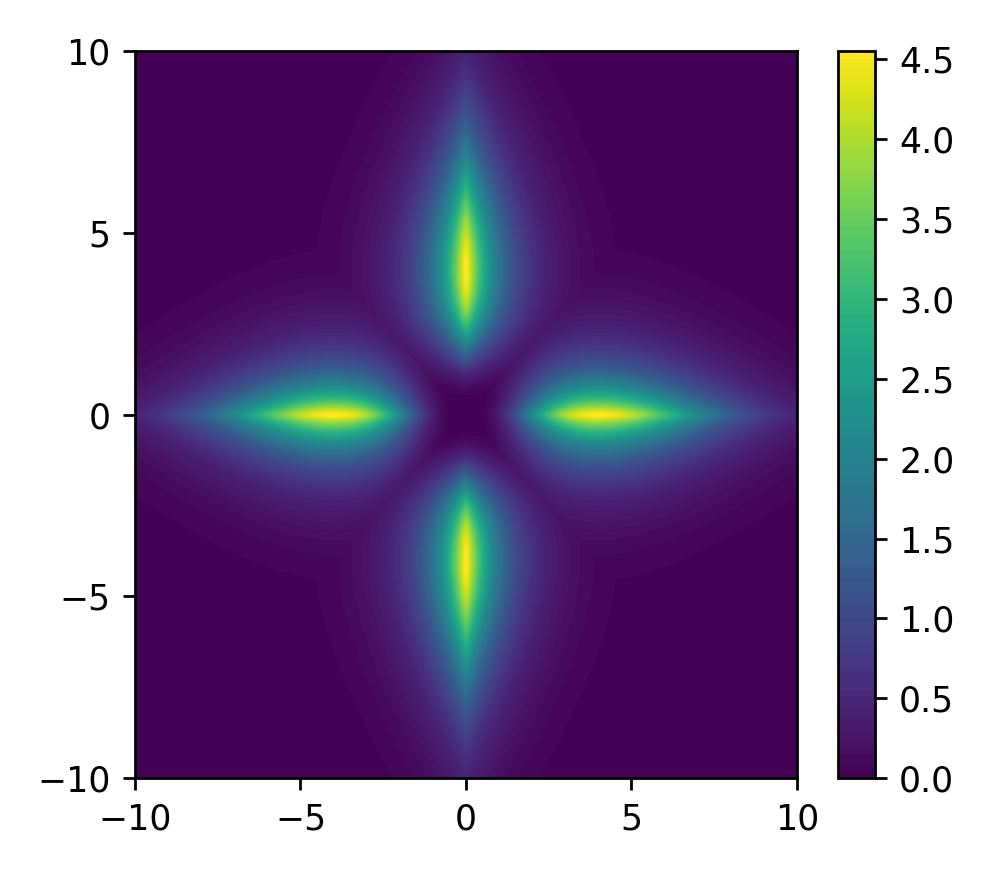}
	\end{center}
	\caption{Illustration of the axial modes target density \eqref{Eq:axial_modes_target} in dimension $d=2$. \label{Fig:axial_modes_target}}
\end{figure}

\section{Formal Derivation} \label{Sec:Deri}

In order to explain how PSS as described in Section \ref{Sec:Intro} can be derived from the usual slice sampling framework laid out in \cite{BesagGreen}, we provide two technical tools. The first one is a well-known identity from measure theory.

\begin{proposition} \label{Prop:polar_coords}
	Any integrable real-valued function $g$ on $(\R^d,\B(\R^d))$ satisfies
	\begin{equation*}
		\int_{\R^d} g(x) \d x
		= \int_0^{\infty} \int_{\sph^{d-1}} g(r \theta) r^{d-1} \sigma_d(\d \theta) \d r ,
	\end{equation*}
	where $\sigma_d$ is the surface measure on $(\sph^{d-1},\B(\sph^{d-1}))$.
\end{proposition}
\begin{proof}
	See for example Theorem 15.13 in \cite{Schilling}.
\end{proof}

Though this identity does not by itself involve any probabilistic quantities, we apply it in a stochastic setting to obtain our second tool, which we term \textit{sampling in polar coordinates}.

\begin{corollary} \label{Cor:SIPC}
	Let $\xi$ be a distribution on $(\R^d, \B(\R^d))$ with probability density $\varrho_{\xi}$. Then by Proposition \ref{Prop:polar_coords} we get for any $A \in \B(\R^d)$ that
	\begin{align*}
		\xi(A) 
		&= \int_{\R^d} \ind_A(x) \varrho_{\xi}(x) \d x \\
		&= \int_{\sph^{d-1}} \int_0^{\infty} \ind_A(r \theta) \varrho_{\xi}(r \theta) r^{d-1} \d r \sigma_d(\d \theta) .
	\end{align*}
	Consequently a random variable $X \sim \xi$ can be sampled in polar coordinates as $X := R \cdot \Theta$ by sampling $(R, \Theta)$ from the joint distribution with probability density
	\begin{equation*}
		(r, \theta) \mapsto \varrho_{\xi}(r \theta) r^{d-1} \ind_{\R_+}(r)
	\end{equation*}
	w.r.t.~$\lambda_1 \otimes \sigma_d$, where $\lambda_1$ denotes the Lebesgue measure on $(\R, \B(\R))$.
\end{corollary}

We can now apply the second tool to PSS. In the framework of \citet{BesagGreen}, the $X$-update of slice sampling for a target density factorized as in \eqref{Eq:den_fac} is to be performed by sampling $X_n$ from the distribution with unnormalized density
\begin{equation*}
	x \mapsto \varrho_{\nu}^{(0)}(x) \ind_{\!\ooint{t_n}{\infty}}(\varrho_{\nu}^{(1)}(x)) 
	= \norm{x}^{1-d}\ind_{\!\ooint{t_n}{\infty}}(\varrho_{\nu}^{(1)}(x)) .
\end{equation*}
By Corollary \ref{Cor:SIPC}, this can equivalently be done by sampling $X_n$ in polar coordinates as $X_n := R_n \Theta_n$, where $(R_n, \Theta_n)$ is drawn from the joint distribution with unnormalized density
\begin{align*}
	(r, \theta) 
	\mapsto \; &r^{1-d} \ind_{\!\ooint{t_n}{\infty}}(\varrho_{\nu}^{(1)}(r\theta))\, r^{d-1} \ind_{\R_+}(r) \\
	&= \ind_{\!\ooint{t_n}{\infty}}(\varrho_{\nu}^{(1)}(r\theta)) \ind_{\R_+}(r) 
\end{align*}
w.r.t.~$\lambda_1 \otimes \sigma_d$. As this distribution is simply \eqref{Eq:PSS_Xup}, PSS as described in Section \ref{Sec:Intro} is equivalent to the slice sampler resulting from factorization \eqref{Eq:den_fac} in the general slice sampling framework of \citet{BesagGreen}.

\section{Proofs} \label{Sec:Proofs}

We assume some familiarity with transition kernels and Markov chains throughout this section. For details we refer to the introductory sections of \cite{MCbook}.

We introduce some notation and provide a few observations. Let $C_{\nu} := \int_{\R^d} \varrho_{\nu}(x) \d x$, so that
\begin{equation*}
	C_{\nu} \, \nu(\d x) = \varrho_{\nu}^{(0)}(x) \varrho_{\nu}^{(1)}(x) \d x .
\end{equation*}
For $t>0$ define the level set
\begin{equation*}
	L(t)
	:= \{x \in \R^d \mid \varrho_{\nu}^{(1)}(x) > t\}.
\end{equation*}
Note that, by
\begin{equation*}
	\ind_{\ooint{0}{\varrho_{\nu}^{(1)}(x)}}(t) = \ind_{L(t)}(x) ,
\end{equation*}
the transition kernel corresponding to PSS for $\nu$ can be expressed as
\begin{equation*}
	P(x,A)
	:= \frac{1}{\varrho_{\nu}^{(1)}(x)} \int_0^{\infty} \mu_t(A) \ind_{L(t)}(x) \d t 
\end{equation*}
for $x \in \R^d, A \in \B(\R^d)$, where we set
\begin{equation*}
	\mu_t(A)
	:= \frac{\int_A \varrho_{\nu}^{(0)}(x) \ind_{L(t)}(x) \d x}{\int_{\R^d} \varrho_{\nu}^{(0)}(x) \ind_{L(t)}(x) \d x}
\end{equation*}
for $t > 0$. It is known that $\nu P = \nu$ \cite{PolarSS}. We formulate a criterion for invariance w.r.t.~imitations of polar slice sampling.

\begin{lemma} \label{Lem:HSS_inv}
	Suppose that for each $t > 0$ we have a transition kernel $U_X^{(t)}$ on $\R^d \times \B(\R^d)$ with $\mu_t U_X^{(t)} = \mu_t$. Then the transition kernel $Q$ on $\R^d \times \B(\R^d)$ given by
	\begin{equation*}
		Q(x,A)
		:= \frac{1}{\varrho_{\nu}^{(1)}(x)} \int_0^{\infty} U_X^{(t)}(x,A) \ind_{L(t)}(x) \d t 
	\end{equation*}
	satisfies $\nu Q = \nu$.
\end{lemma}
\begin{proof}
	Let $A \in \B(\R^d)$ arbitrary. Then
	\begin{align*}
		&C_{\nu} \cdot \nu Q (A) 
		= C_{\nu} \int_{\R^d} Q(x,A) \nu(\d x) \\ 
		&= \int_{\R^d} \int_0^{\infty} U_X^{(t)}(x,A) \ind_{L(t)}(x) \d t \varrho_{\nu}^{(0)}(x) \d x \\ 
		&= \int_0^{\infty} \int_{\R^d} U_X^{(t)}(x,A) \varrho_{\nu}^{(0)}(x) \ind_{L(t)}(x) \d x \d t \\ 
		&= \int_0^{\infty} \left( \int_{\R^d} U_X^{(t)}(x,A) \mu_t(\d x) \right) \\
		&\qquad\;\;\, \cdot \left( \int_{\R^d} \varrho_{\nu}^{(0)}(x) \ind_{L(t)}(x) \d x \right) \d t \\ 
		&= \int_0^{\infty} \mu_t(A) \int_{\R^d} \varrho_{\nu}^{(0)}(x) \ind_{L(t)}(x) \d x \d t \\ 
		&= \int_{\R^d} \int_0^{\infty} \mu_t(A) \ind_{L(t)}(x) \d t \varrho_{\nu}^{(0)}(x) \d x \\ 
		&= C_{\nu} \int_{\R^d} P(x,A) \nu(\d x) \\ 
		&= C_{\nu} \cdot \nu P (A) = C_{\nu} \cdot \nu(A) , 
	\end{align*}
	which shows $\nu Q = \nu$.
\end{proof}

Note that the framework we rely on was previously used in Lemma~1 in \cite{Latuszynski} to analyze the reversibility of hybrid uniform slice sampling. Further note that in proving Lemma \ref{Lem:HSS_inv} we never use the precise factorization of $\varrho_{\nu}$ into $\varrho_{\nu}^{(0)}$ and $\varrho_{\nu}^{(1)}$ that is dictated by PSS. Hence the result also holds true for any other slice sampling scheme that adheres to this format.
Moreover, in the previous lemma it is assumed that $U_X^{(t)}$ is a transition kernel on $\R^d\times \B(\R^d)$. For given $t>0$ it is sometimes more convenient to define a corresponding transition kernel $U_X^{(t)}$ on $L(t)\times \B(L(t))$. In such cases we consider $\overline{U}_X^{(t)}$ as extension of $U_X^{(t)}$ to $\R^d\times \B(\R^d)$ given as 
\begin{equation*}
	\overline{U}_X^{(t)}(x,A) = 
	\begin{cases}
		U^{(t)}_X(x,A\cap L(t)) & x\in L(t) \\
		\ind_A(x) & x\not\in L(t) .
	\end{cases}	
\end{equation*}
In the following we write $U_X^{(t)}$ for $\overline{U}_X^{(t)}$ and consider $U_X^{(t)}$ as extension if necessary.
We turn our attention to the 1st variant of GPSS.
\begin{lemma} \label{Lem:GPSS1_update_inv}
	For $T_n=t$ the transition kernel $U^{(t)}_{X}$ on $(L(t), \B(L(t)))$ of the $X$-update of the 1st variant of GPSS takes the form $U^{(t)}_X = U_{D_1}^{(t)} U_R^{(t)}$ with kernels
	\begin{align}
		U_{D_1}^{(t)}(r\theta,A) & = \frac{\int_{\sph^{d-1}} \ind_A(r \tilde{\theta}) \sigma_d(\d \tilde{\theta})}{\int_{\sph^{d-1}} \ind_{L(t)}(r \tilde{\theta}) \sigma_d(\d \tilde{\theta})}, \notag \\
		U_R^{(t)}(r\theta,A) & = \frac{\int_0^{\infty} \ind_A(\tilde{r} \theta) \d \tilde{r}}{\int_0^{\infty} \ind_{L(t)}(\tilde{r} \theta) \d \tilde{r}}, 
		\label{al: radius_update_kernel}
	\end{align}
	where $r \in \R_+, \theta \in \sph^{d-1}, x=r \theta \in L(t)$ and $A \in \B(L(t))$. Moreover, $U_{D_1}^{(t)}$ and $U_R^{(t)}$ are reversible w.r.t.~$\mu_t$.
\end{lemma}
\begin{proof}
	The overall $X$-update consists of consecutively realizing \eqref{Eq:GPSS_dir} and \eqref{Eq:GPSS_rad}. Obviously, the direction update $r_{n-1} \theta_{n-1} \mapsto r_{n-1} \theta_n$ (according to \eqref{Eq:GPSS_dir}) corresponds to $U_{D_1}^{(t)}$ and the radius update $r_{n-1} \theta_n \mapsto r_n \theta_n$ (according to \eqref{Eq:GPSS_rad}) corresponds to $U_R^{(t)}$. Consequently, the transition kernel of the $X$-update is $U^{(t)}_X = U_{D_1}^{(t)} U_R^{(t)}$, that is, $U^{(t)}_X$ is the product of the kernels $U_{D_1}^{(t)}$ and $U_R^{(t)}$.
	
	Now we prove the claimed invariance property. We use the fact that $\varrho_{\nu}^{(0)}(x) = \norm{x}^{1-d}$. To simplify the notation, set
	\begin{equation}
		c_t := \left(\int_{\R^d} \norm{x}^{1-d} \ind_{L(t)}(x) \d x\right)^{-1} ,
		\label{al:normalizing_constant}
	\end{equation}
	so that the restriction of $\mu_t$ to $L(t)$ can be written as 
	\begin{equation*}
		\mu_t(\d x)
		= c_t \norm{x}^{1-d} \d x .
	\end{equation*}
	Observe that Proposition \ref{Prop:polar_coords} yields for all $A, B \in \B(L(t))$
	\begin{align*}
		&\int_A U_{D_1}^{(t)}(x, B) \mu_t(\d x) \\
		&= c_t \int_A U_{D_1}^{(t)}(x, B) \norm{x}^{1-d} \d x \\ 
		&= c_t \int_0^{\infty} \int_{\sph^{d-1}} \ind_A(r \theta) U_{D_1}^{(t)}(r \theta, B) \sigma_d(\d \theta) \d r \\ 
		&= c_t \int_0^{\infty} \frac{\int_{\sph^{d-1}} \ind_A(r \theta) \sigma_d(\d \theta) \int_{\sph^{d-1}} \ind_B(r \tilde{\theta}) \sigma_d(\d \tilde{\theta})}{\int_{\sph^{d-1}} \ind_{L(t)}(r \tilde{\theta}) \sigma_d(\d \tilde{\theta})} \d r. 
	\end{align*}
	Similarly, again by Proposition \ref{Prop:polar_coords}, we obtain
	\begin{align*}
		&\int_A U_R^{(t)}(x, B) \mu_t(\d x) \\
		&= c_t \int_A U_R^{(t)}(x, B) \norm{x}^{1-d} \d x \\ 
		&= c_t \int_{\sph^{d-1}} \int_0^{\infty} \ind_A(r \theta) U_R^{(t)}(r \theta, B) \d r \sigma_d(\d \theta) \\ 
		&= c_t \int_{\sph^{d-1}} \frac{\int_0^{\infty} \ind_A(r \theta) \d r \int_0^{\infty} \ind_B(\tilde{r} \theta) \d \tilde{r}}{\int_0^{\infty} \ind_{L(t)}(\tilde{r} \theta) \d \tilde{r}} \sigma_d(\d \theta) . 
	\end{align*}
	As the last expression in each of these two computations is symmetric in $A$ and $B$, they show both $U_{D_1}^{(t)}$ and $U_R^{(t)}$ to be reversible w.r.t.~$\mu_t$. 
\end{proof}

We also provide a suitable representation for the 2nd variant of GPSS.

\begin{lemma} \label{Lem:GPSS2_update_inv}
	For $T_n=t$ the transition kernel $U^{(t)}_{X}$ on $(L(t), \B(L(t)))$ of the $X$-update of the 2nd variant of GPSS takes the form $U^{(t)}_X = U_{D_2}^{(t)} U_R^{(t)}$ with $U_R^{(t)}$ being specified as in \eqref{al: radius_update_kernel} and a suitable\footnote{We provide an explicit expression of $U_{D_2}^{(t)}$ in the proof of the statement.}, w.r.t.~$\mu_t$ reversible, kernel $U_{D_2}^{(t)}$.
\end{lemma}

\begin{proof}
	We require some further notation: For $r\in\R_+$ and $\theta,\vartheta\in \sph^{d-1}$ let
	$g^{(\theta,\vartheta)} \colon [0,2\pi] \to \sph^{d-1}$ be given by $g^{(\theta,\vartheta)}(\alpha) = \theta \cos(\alpha) + \vartheta \sin(\alpha)$ and define
	\begin{align*}
		L(t,r) & := \{ \theta \in \sph^{d-1} \mid r \theta \in L(t) \}	\\
		L(t,r,\theta,\vartheta) & := \{ \alpha\in [0,2\pi] \mid g^{(\theta,\vartheta)}(\alpha)\in L(t,r) \}.
	\end{align*}
	Define the transition kernel $S^{(t,r)}$ on $L(t,r)\times \B(\sph^{d-1})$ as
	\begin{equation*}
		S^{(r,t)}(\theta,C) := \int_{\sph^{d-2}_{\theta}} \int_{L(t,r,\theta,\vartheta)} \frac{\ind_C(g^{(\theta,\vartheta)}(\alpha))\,\d \alpha}{\lambda_1(L(t,r,\theta,\vartheta))} \xi_{\theta}(\d \vartheta),
	\end{equation*}
	where $\theta\in L(t,r)$ and $C\in\B(\sph^{d-1})$, with $\xi_{\theta}$ being the uniform distribution on $\sph^{d-2}_{\theta}$. 
	The transition kernel $S^{(r,t)}$ coincides with the transition kernel of the ideal geodesic slice sampler on the sphere that is reversible w.r.t.~the uniform distribution on $L(t,r)$, see Lemma~14 in \cite{SphericalSS}. Denote the uniform distribution on $L(t,r)$ as $\gamma^{(r,t)}$, i.e.
	\begin{equation*}
		\gamma^{(r,t)}(\d \theta) = \frac{ \ind_{L(t,r)}(\theta) \sigma_d(\d \theta)}{\sigma_d(L(t,r))}.
	\end{equation*} 
	Now observe that the transition kernel of the direction update $U_{D_2}^{(t)}$ specified in the 2nd variant of GPSS for $x = r \theta\in L(t)$ with $r\in\R_+$ and $\theta\in \sph^{d-1}$ is given by
	\begin{equation}
		U^{(t)}_{D_2}(r\theta,A) = S^{(r,t)}(\theta,A^{(r)}), 
		\label{Eq:repre_D_2}
	\end{equation}
	where for $A\in \B(\R^d)$ we denote 
	\begin{equation*}
		A^{(r)} := \{ \vartheta\in \sph^{d-1} \mid r\vartheta \in A \}.
	\end{equation*}
	Note that the radius update of the 2nd variant of GPSS coincides with the one of the 1st variant, i.e. $U_R^{(t)}$ is given by \eqref{al: radius_update_kernel}. Hence the total $X$-update takes the form $U^{(t)}_X = U^{(t)}_{D_2} U_R^{(t)}$. It remains to prove the reversibility of $U_{D_2}^{(t)}$ w.r.t.~$\mu_t$: For $A, B \in \B(L(t))$ we have with $c_t$ as in \eqref{al:normalizing_constant} that
	\begin{align*}
		&\int_A U_{D_2}^{(t)}(x, B) \mu_t(\d x) \\
		&= c_t \int_A U_{D_2}^{(t)}(x, B) \norm{x}^{1-d} \d x \\ 
		&= c_t \int_0^{\infty} \int_{\sph^{d-1}} \ind_A(r \theta) U_{D_2}^{(t)}(r \theta, B) \sigma_d(\d \theta) \d r \\ 
		&= c_t \int_0^{\infty} \int_{A^{(r)}} S^{(r,t)}(\theta,B^{(r)}) \gamma^{(r,t)}(\d \theta) \cdot \sigma_d(L(t,r)) \,\d r, 
	\end{align*}
	which is, by the reversibility of $S^{(r,t)}$ w.r.t.~$\gamma^{(r,t)}$, symmetric in $A$ and $B$. Consequently, $U_{D_2}^{(t)}$ is reversible w.r.t.~$\mu_t$ and the claimed statement is proven.
\end{proof}

Theorem \ref{Thm:GPSS_inv} is now an easy consequence.

\begin{proof}[Proof of Theorem \ref{Thm:GPSS_inv}]
	By Lemmas~\ref{Lem:GPSS1_update_inv} and \ref{Lem:GPSS2_update_inv}, given $T_n=t$, the transition kernel of the $X$-update of the 1st and 2nd variant of GPSS is $U_X^{(t)} := U_{D_i}^{(t)} U_R^{(t)}$ for $i=1,2$ respectively. By the reversibility of the individual kernels w.r.t.~$\mu_t$, see Lemma~\ref{Lem:GPSS1_update_inv} and Lemma~\ref{Lem:GPSS2_update_inv}, we have
	\begin{equation*}
		\mu_t U_X^{(t)} 
		= \mu_t U_{D_i}^{(t)} U_R^{(t)}
		= \mu_t U_R^{(t)}
		= \mu_t ,
	\end{equation*}
	proving that $U_X^{(t)}$ leaves $\mu_t$ for $i=1,2$ invariant. By Lemma \ref{Lem:HSS_inv}, this yields that the variants of GPSS leave the target distribution $\nu$ invariant.
\end{proof}

We add another auxiliary result. 
\begin{lemma} \label{Lem:almost_surely_finite}
	For $(t,\theta)\in \ooint{0}{\infty} \times \sph^{d-1}$ let
	\begin{equation*}
		L(t,\theta) := \{ r \in \R_+ \mid r\theta \in L(t) \} .
	\end{equation*}
	Then $\lambda_1(L(t,y/\norm{y}))<\infty$ holds for $\lambda_1\otimes \lambda_d$-almost all $(t,y)\in \R_+\times \R^d$.
\end{lemma}
\begin{proof}
	By Proposition \ref{Prop:polar_coords} follows
	\begin{align*}
		\int_{\R^d} \varrho_{\nu}(x) \d x
		&= \int_{\sph^{d-1}} \int_0^\infty \varrho_{\nu}^{(1)}(r\theta) \d r \sigma_d(\d \theta) \\
		&= \int_0^\infty \int_{\sph^{d-1}} \int_0^\infty \ind_{L(t)}(r \theta) \d r \sigma_d(\d \theta) \d t \\
		&= \int_0^\infty \int_{\sph^{d-1}} \lambda_1(L(t,\theta)) \sigma_d(\d \theta) \d t .
	\end{align*}
	From this and the fact that $\varrho_{\nu}$ is integrable by assumption, it is immediate that $\lambda_1(L(t,\theta)) < \infty$ for $\lambda_1 \otimes \sigma_d$-almost all $(t,\theta) \in \R_+ \times \sph^{d-1}$. By defining
	\begin{equation*}
		F := \{(t,\theta) \in \R_+ \times \sph^{d-1} \mid \lambda_1(L(t,\theta)) = \infty \} ,
	\end{equation*}
	we can alternatively express this result as $(\lambda_1 \otimes \sigma_d)(F) = 0$. With another application of Proposition \ref{Prop:polar_coords}, this yields
	\begin{align*}
		&(\lambda_1 \otimes \lambda_d)(\{(t,y) \in \R_+ \times \R^d \mid \lambda_1(L(t, y/\norm{y})) = \infty\}) \\
		&= \; (\lambda_1 \otimes \lambda_d)(\{(t,y) \in \R_+ \times \R^d \mid (t, y/\norm{y}) \in F \}) \\
		&= \int_{\R^d} \int_0^{\infty} \ind_F(t, y/\norm{y}) \d t \d y \\
		&= \int_0^{\infty} \int_{\sph^{d-1}} \int_0^{\infty} \ind_F(t, \theta) \d t \, \sigma_d(\d \theta) r^{d-1} \d r \\
		&= \int_0^{\infty} (\lambda_1 \otimes \sigma_d)(F) \, r^{d-1} \d r 
		= 0 ,
	\end{align*}
	which proves the lemma's claim.
\end{proof}

Now we turn to the proof of Theorem~\ref{Thm:GPSS_variant2_conv}.

\begin{proof}[Proof of Theorem \ref{Thm:GPSS_variant2_conv}]
	We use the same notation as in the proof of Theorem~\ref{Thm:GPSS_inv}.
	Moreover, we know from the proof\footnote{The theorem applies in their notation with $p(\theta)=\ind_{L(t,r)}(\theta)$, $C=L(t,r)$ and $\beta=1$, where also Remark~1 of \cite{SphericalSS} should be taken into account.} of Theorem~15 in \cite{SphericalSS} that there exists a constant $\varepsilon>0$ (independent of $t,r$) such that
	\begin{equation*}
		S^{(r,t)}(\theta,D) \geq \varepsilon \, \sigma_d(D\cap L(t,r)) ,
	\end{equation*}
	for any $\theta\in L(t,r)$ and $D\in\B(\sph^{d-1})$.
	By \eqref{Eq:repre_D_2} this implies for $x=r\theta\in L(t)$ and $B\in\B(L(t))$ that
	\begin{align*}
		& U_{D_2}^{(t)}(r\theta,B) 
		= S^{(r,t)}(\theta,B^{(r)}) \geq \varepsilon \, \sigma_d(B^{(r)}\cap L(t,r))\\
		& = \varepsilon \int_{\sph^{d-1}} \ind_{B\cap L(t)}(r\vartheta) \sigma_d(\d \vartheta)\\
		& = \varepsilon \int_{\sph^{d-1}} \ind_{L(t)}(r\vartheta) \delta_{r\vartheta}(B) \sigma_d(\d \vartheta),
	\end{align*}
	where $\delta_z$ denotes the Dirac-measure at $z\in\R^d$. Concisely written down, the former inequality yields
	\begin{equation*}
		U_{D_2}^{(t)}(r\theta,\d y) \geq \varepsilon \int_{\sph^{d-1}} \ind_{L(t)}(r \vartheta) \delta_{r\vartheta}(\d y) \sigma_d(\d \vartheta) .
	\end{equation*}
	For $\vartheta\in\sph^{d-1}$ we use $L(t,\vartheta)$ as defined in Lemma~\ref{Lem:almost_surely_finite} and note that the normalizing constant within $U_R^{(t)}$ satisfies $\int_{0}^{\infty} \ind_{L(t)}(s\vartheta) \d s = \lambda_1(L(t,\vartheta))$.
	Taking the representation of the $X$-update of the 2nd variant of GPSS into account we obtain (using the same variables as earlier)
	\begin{align*}
		& U_X^{(t)}(r \theta,B) = U_{D_2}^{(t)} U_R^{(t)}(r\theta,B) \\
		& = \int_{\R^d} U_R^{(t)}(y,B) U_{D_2}^{(t)}(r\theta,\d y) \\
		& \geq \varepsilon \int_{\sph^{d-1}} \int_{\R^d} U_R^{(t)}(y,B) \ind_{L(t)}(r\vartheta) \delta_{r\vartheta}(\d y) \sigma_d(\d \vartheta) \\
		& = \varepsilon \int_{\sph^{d-1}} U_R^{(t)}(r\vartheta,B) \ind_{L(t)}(r\vartheta) \sigma_d(\d \vartheta) \\
		& = \varepsilon \int_{\sph^{d-1}} \int_{0}^{\infty} \frac{\ind_{B} (\tilde{r} \vartheta)} {\lambda_1(L(t,\vartheta))} \ind_{L(t)}(r\vartheta) \d\tilde{r}\, \sigma_d(\d \vartheta) \\
		& = \varepsilon \int_B \frac{\ind_{L(t)}(ry/\norm{y})}{\norm{y}^{d-1} \lambda_1(L(t,y/\norm{y}))} \d y ,
	\end{align*}
	where the last equality follows by Proposition~\ref{Prop:polar_coords}. Therefore the transition kernel $Q$ corresponding to the 2nd variant of GPSS satisfies for any $A \in \B(\R^d)$ and $0 \neq x \in \R^d$ with $x = r\theta$ that
	\begin{align*}
		& Q(x,A) = 
		\frac{1}{\varrho_{\nu}^{(1)}(x)} \int_0^{\varrho_{\nu}^{(1)}(x)} U_X^{(t)}(x,A\cap L(t)) \d t \\
		& \geq \frac{\varepsilon}{\varrho_{\nu}^{(1)}(x)} \int_A \int_0^{\varrho_{\nu}^{(1)}(x)} \frac{\ind_{L(t)}(ry/\norm{y}) \ind_{L(t)}(y)}{\norm{y}^{d-1} \lambda_1(L(t,y/\norm{y}))} \d t\, \d y \\
		& = \frac{\varepsilon}{\varrho_{\nu}^{(1)}(x)} \int_A \int_0^{\min\{\varrho_{\nu}^{(1)}(x),\varrho_{\nu}^{(1)}(y),\varrho_{\nu}^{(1)}(ry/\norm{y})\}} \\
		& \qquad \qquad\qquad\quad \cdot \frac{\norm{y}^{1-d} }{\lambda_1(L(t,y/\norm{y}))} \d t\, \d y.
	\end{align*}
	By Lemma~\ref{Lem:almost_surely_finite} we have that the mapping 
	\begin{equation*}
		(t,y) \mapsto \lambda_1(L(t,y/\norm{y}))
	\end{equation*}
	is $\lambda_1\otimes \lambda_d$-almost surely finite, so that we obtain that the function
	\begin{equation*}
		(t,y) \mapsto \frac{\norm{y}^{1-d} }{\lambda_1(L(t,y/\norm{y}))} 
	\end{equation*}
	is $\lambda_1\otimes \lambda_d$-almost surely strictly larger than zero. By the assumption and definition of $\varrho_{\nu}^{(1)}$ we have that $\varrho_{\nu}^{(1)} > 0$ on $\R^d \setminus\{0\}$, such that $Q(x,A) > 0$ whenever $\lambda_d(A) > 0$. We apply the former implication: Let $A\in\B(\R^d)$ such that $\nu(A) > 0$. Then, by the absolute continuity of $\nu$ w.r.t.~$\lambda_d$ we obtain that $\lambda_d(A) > 0$ and thus $Q(x,A)>0$ for any $0 \neq x\in \R^d$. This yields that the transition kernel of the 2nd variant of GPSS is $\nu$-irreducible and aperiodic, as defined in Section~3 in \cite{Tierney}.
	
	Since by Theorem~\ref{Thm:GPSS_inv} we also know that $\nu$ is an invariant distribution of $Q$, applying Theorem~1 of \cite{Tierney} yields that $\nu$ is the unique invariant distribution. Moreover, the same theorem gives that the distribution of $X_n$ converges $\nu$-almost surely (regarding the initial state) to $\nu$ in the total variation distance. 
\end{proof}

\section{Implementation Notes} \label{Sec:ImpNotes}

The code we provide not only allows for the reproduction of our experimental results, but also contains an easily usable, general purpose implementation of GPSS in Python 3.10, based on numpy. 

Those still seeking to implement GPSS themselves, perhaps in another programming language, can pretty much follow Algorithms \ref{Alg:GPSS}, \ref{Alg:OS} and \ref{Alg:RS}. Still we find it appropriate to give two pieces of advice. First, the target density $\varrho_{\nu}$, its transform $\varrho_{\nu}^{(1)}$ and the thresholds $t_n$ should all be moved into log-space in order to make the implementation numerically stable. Second, upon generating a direction proposal, i.e.~immediately following the code-adaptation of line 6 of Algorithm \ref{Alg:OS}, the proposal should be re-normalized. 

Although in theory the proposal should already be normalized, in practice the operations involved in generating it always introduce small numerical errors. Individually, these errors are far too small to matter, but, due to the way in which each direction variable depends on that from the previous iteration, without re-normalization the errors accumulate and lead the direction variables to be more and more unnormalized. If the sampler goes through many thousands of iterations, these accumulating errors eventually introduce large amounts of bias into the sampling, because the initial proposal set in the direction update is no longer a great circle but rather an elongated ellipse. Of course the re-normalization is not necessary if the direction variables are extracted from the full samples whenever required instead of being stored separately.

\section{Additional Sampling Statistics} \label{Sec:SamStats}

\begin{table}[h]
	\caption{Target density evaluations per iteration (TDE/I) and integrated autocorrelation times (IAT) for the experiments presented in Sections \ref{SubSec:std_Cauchy} and \ref{SubSec:hyperplane}. The IATs were computed w.r.t.~the sample log radii in the Cauchy experiment and w.r.t.~the sample radii in the hyperplane experiment.}
	\label{Tab:tdes_iats}
	\vskip 0.15in
	\begin{center}
		\begin{small}
			\begin{sc}
				\begin{tabular}{lcccr}
					\toprule
					& \multicolumn{2}{c}{Cauchy} & \multicolumn{2}{c}{Hyperplane} \\
					Sampler & TDE/I & IAT & TDE/I & IAT \\
					\midrule
					GPSS & 6.90 & 8.59 & 12.23 & 1.09 \\
					HRUSS & 8.46 & 51346.93 & 7.78 & 684.57 \\
					naive ESS & 5.86 & 35543.94 & 7.91 & 199.17 \\
					tuned ESS & -- & -- & 6.51 & 188.13 \\
					\bottomrule
				\end{tabular}
			\end{sc}
		\end{small}
	\end{center}
	\vskip -0.1in
\end{table}

\section{Bayesian Logistic Regression} \label{Sec:LogReg}

In this section we examine how well GPSS performs in a more classical machine learning setting, namely Bayesian logistic regression with a mean-zero Gaussian prior. We begin by giving a brief overview of the problem. Bayesian logistic regression is a binary regression problem. Supposing the training data to be given as pairs $(a^{(i)}, b^{(i)})$, $i=1,...,m$, where $a^{(i)} \in \R^d$ and $b^{(i)} \in \{-1,1\}$, and choosing the prior distribution as $\Nc_d(0, 0.1^2 I_d)$, the posterior density of interest is given by
\begin{equation*}
	\varrho_{\pi}(x)
	= \Nc_d(x; 0, 0.1^2 I_d) \prod_{i=1}^m \frac{1}{1 + \exp(-b^{(i)} \langle a^{(i)}, x \rangle)} ,
\end{equation*}
where $\langle \cdot, \cdot \rangle$ denotes the standard inner product on $\R^d$.

We consider the \textit{CoverType data set} \cite{CoverTypeThesis,CoverTypeURL}, which has cartographic variables as its features and forest cover types as its labels. The raw data set has $m = $ 581,012 instances, $54$ features and $7$ different labels. As one of the labels occurs in about the same number of instances (283,301) as all the other labels combined (297,711), we make the regression problem binary by mapping this label to $+1$ and all other labels to $-1$. As the data set is quite large and we want to keep all of our experiments easily reproducible (in particular avoiding prohibitively large execution times), we use only $10\%$ of it as training data and the remaining $90\%$ as test data. Moreover, we normalized the data (i.e.~each feature was affine-linearly transformed to have mean zero and standard deviation one) and added a constant feature (to enable a non-zero intercept of the regression line), rendering the problem $d=55$ dimensional.

We then ran the usual three slice samplers and NUTS for $N = 5000$ iterations each, initializing them equally with a draw $x_0$ from the prior distribution $\Nc_d(0, 0.1^2 I_d)$. Viewing the first $N_b = 5000$ iterations as burn-in, we computed sample means for each sampler based only on the latter $N_a = 5000$ iterations and then used these sample means as predictors to compute training and testing accuracies. For reference, we also solved the problem\footnote{Note that sklearn technically solved a slightly different problem, because it only performed maximum likelihood without considering the prior. However, this did not appear to have a significant effect in practice, see Table \ref{Tab:LogReg_accs}.} with the default solver of the classical python machine learning package \texttt{scikit-learn} (typically abbreviated sklearn).

\begin{table}[h]
	\caption{Training and testing accuracies of the different solvers in the Bayesian logistic regression experiment.}
	\label{Tab:LogReg_accs}
	\vskip 0.15in
	\begin{center}
		\begin{small}
			\begin{sc}
				\begin{tabular}{lcr}
					\toprule
					Solver & Train Acc & Test Acc \\
					\midrule
					GPSS & 0.75520 & 0.75571 \\
					HRUSS & 0.75472 & 0.75583 \\
					ESS & 0.75506 & 0.75585 \\
					NUTS & 0.75510 & 0.75573 \\
					sklearn & 0.75562 & 0.75580 \\
					\bottomrule
				\end{tabular}
			\end{sc}
		\end{small}
	\end{center}
	\vskip -0.1in
\end{table}

As can be seen in Table \ref{Tab:LogReg_accs}, all five methods solved the logistic regression problem equally well, achieving virtually identical accuracies on both training and test data. Nevertheless, some differences between the samplers' behaviors become evident when considering the usual sampling metrics, see Table \ref{Tab:LogReg_stats}: On the one hand, GPSS used as many target density evaluations as HRUSS and ESS combined, on the other hand it also achieved more than twice the mean step size of either method (which should in principle help GPSS explore the target distribution's mode more thoroughly, though evidently this did not lead to it finding a better predictor than the other samplers, cf.~Table \ref{Tab:LogReg_accs}). In terms of IAT, the three slice samplers perform relatively similarly\footnote{Note that the IAT values varied significantly between different runs of this experiment.}. Notably, NUTS, for which target density evaluations are neither available nor a sensible metric (due to the sampler's use of gradient information), vastly outperformed all three slice samplers in terms of both IAT and mean step size.

\begin{table}[h]
	\caption{Target density evaluations per iteration (TDE/I), integrated autocorrelation time (IAT) and mean step size (MSS) for the Bayesian logistic regression experiment. The IATs were computed w.r.t.~the marginal samples consisting of the second entry of each sample. Whereas the TDE/I were computed based on the entire sampling runs, both IAT and MSS are only based on the samples generated after the burn-in period.}
	\label{Tab:LogReg_stats}
	\vskip 0.15in
	\begin{center}
		\begin{small}
			\begin{sc}
				\begin{tabular}{lccr}
					\toprule
					Sampler & TDE/I & IAT & MSS \\
					\midrule
					GPSS & 23.55 & 311.37 & 0.0268 \\
					HRUSS & 9.12 & 179.30 & 0.0125 \\
					ESS & 10.81 & 116.17 & 0.0115 \\
					NUTS & -- & 1.18 & 0.2942 \\
					\bottomrule
				\end{tabular}
			\end{sc}
		\end{small}
	\end{center}
	\vskip -0.1in
\end{table}

Aside from classification accuracies, another important aspect to consider when using samplers for logistic regression is how long they need to reach the target distribution's mode from wherever they are initialized. For this we refer to Figures \ref{Fig:LogReg_devs_from_skl} and \ref{Fig:LogReg_radii}, where it can be seen that, at least in this case, GPSS and HRUSS converged towards the mode about equally fast, with the convergence of ESS seemingly being slightly slower. NUTS again worked vastly better than all slice samplers, converging almost instantaneously. This may seem remarkable at first glance, but is to be expected when considering that NUTS -- unlike the slice samplers -- relies on gradient information, which is extremely valuable in solving logistic regression tasks.

\begin{figure}[tb]
	\begin{center}
		\includegraphics[width=0.5\textwidth]{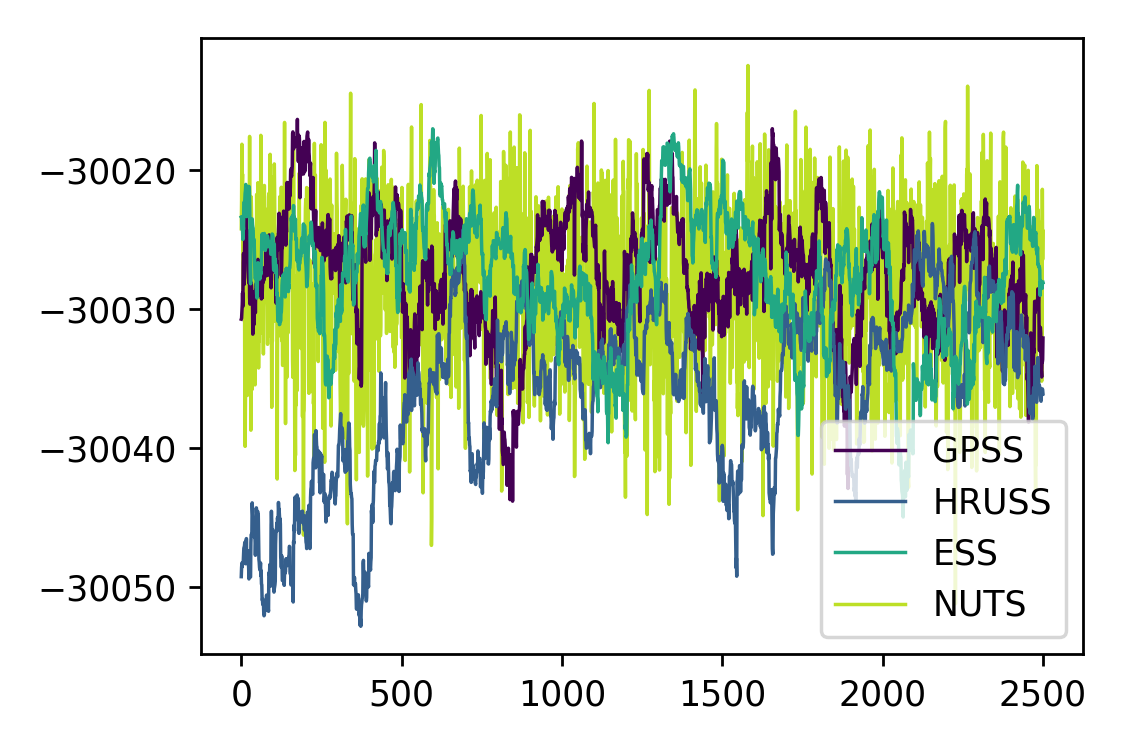}
	\end{center}
	\caption{Log target density values for the Bayesian logistic regression experiment, over the course of each sampler's $N_a = 2500$ iterations after the burn-in period. \label{Fig:LogReg_devs_from_skl}}
\end{figure}

\begin{figure}[tb]
	\begin{center}
		\includegraphics[width=0.5\textwidth]{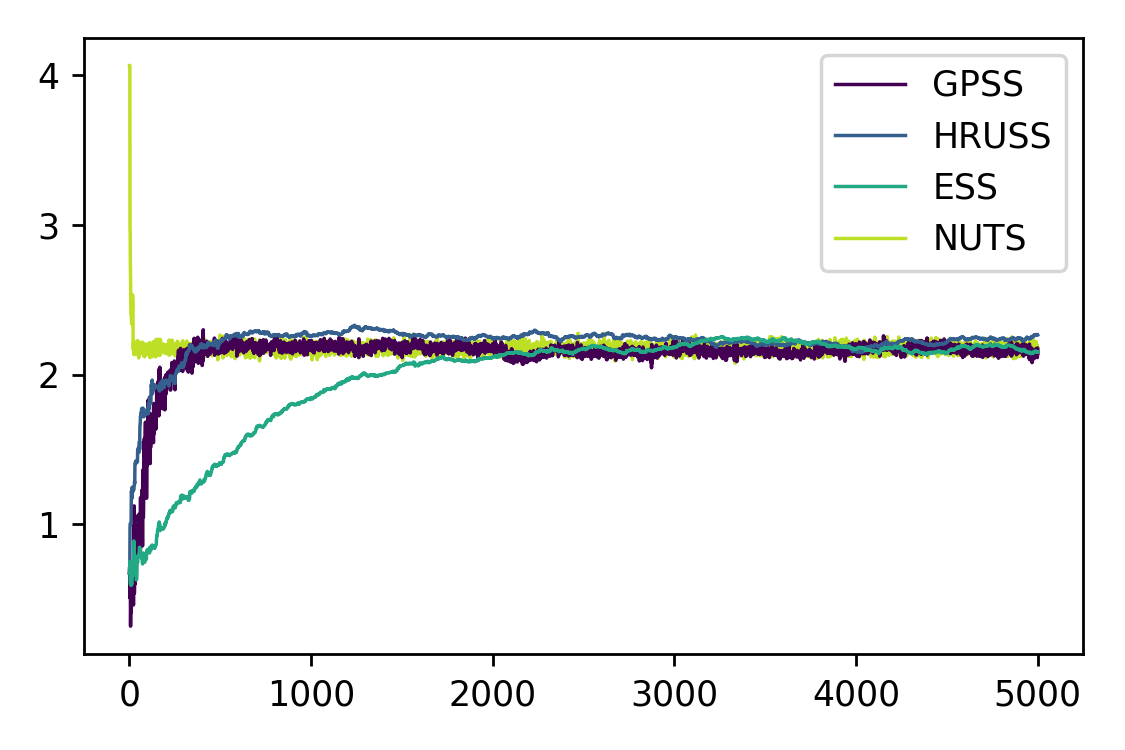}
	\end{center}
	\caption{Sample radii (Euclidean norms) for the Bayesian logistic regression experiment over the course of each sampler's $N = 5000$ iterations. \label{Fig:LogReg_radii}}
\end{figure}

Overall, the results do not suggest Bayesian logistic regression to be a particularly worthwhile application of GPSS, which is not surprising to us based on our observations about the method's apparent strengths and weaknesses (cf.~Sections \ref{Sec:Experiments} and \ref{Sec:Discussion}).

\onecolumn
\section{Additional Results of Numerical Experiments} \label{Sec:ExpRes}
\begin{figure}[H]
	\begin{center}
		\includegraphics[width=\textwidth]{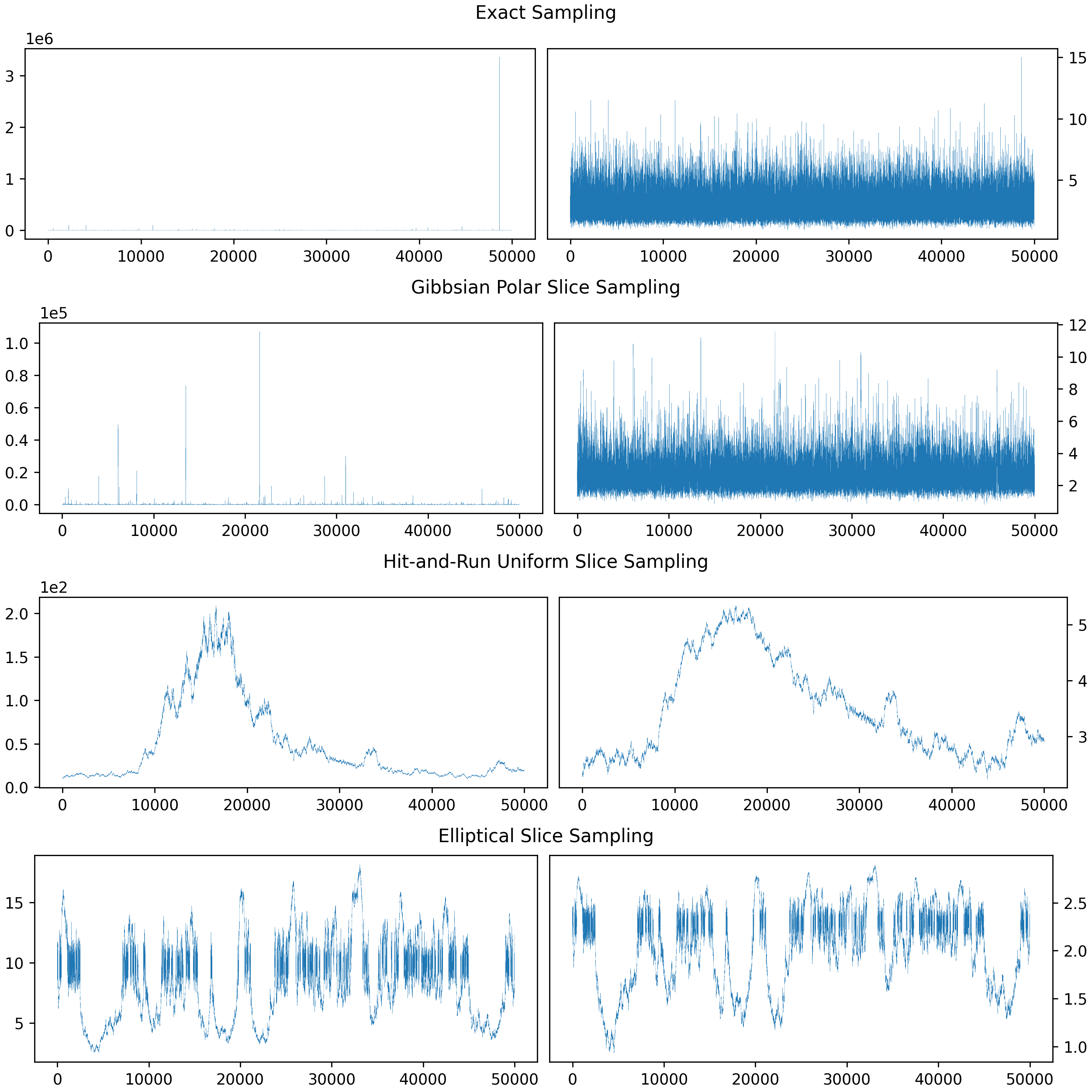}
	\end{center}
	\caption{Sample runs for the multivariate standard Cauchy distribution, i.e.~the density $\varrho_{\nu}(x) = (1 + \norm{x}^2)^{-(d+1)/2}$, in dimension $d=100$. The plots in the left column display the progression of the radii (Euclidean norms) of the individual samples over the course of $N = 5 \cdot 10^4$ iterations. Their counterparts on the right show the logarithms of these values. Note that we include the log radii plots to provide more insight into the short-term behavior of exact sampling and the GPSS chain, which can not be ascertained from the radii plots due to the magnitude of their respective outliers. \label{Fig:std_Cauchy}}
\end{figure}

\begin{figure}[tb]
	\begin{center}
		\includegraphics[width=\textwidth]{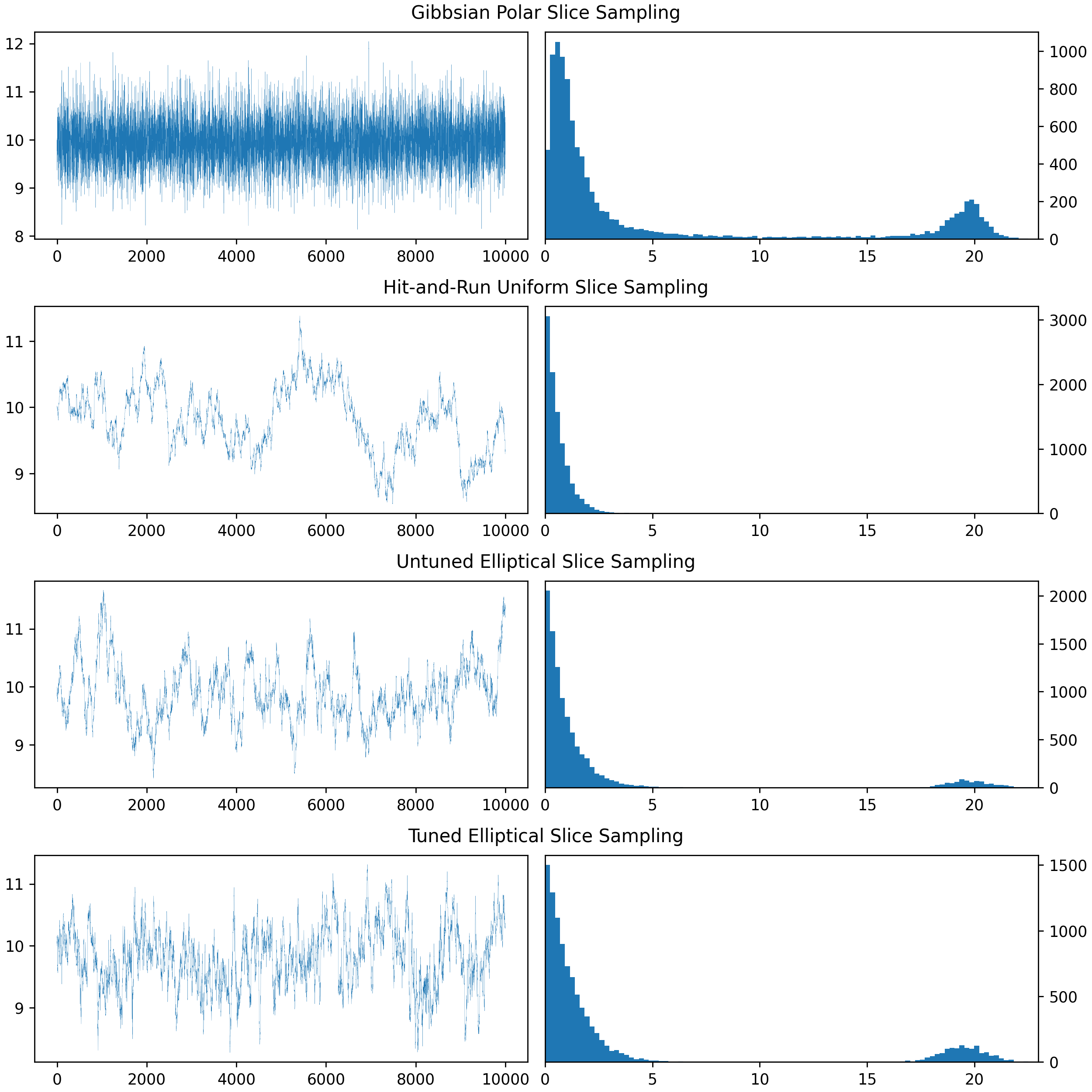}
	\end{center}
	\caption{Sample runs and step size histograms for the hyperplane disk density \eqref{Eq:hyperplane} in dimension $d=200$. The plots in the left column display the progression of the radii (Euclidean norms) of the individual samples over the course of $N = 10^4$ iterations. The plots in the right column show histograms of the Euclidean distances between each two consecutive samples. The descriptor \textit{tuned ESS} refers to ESS where the artificial Gaussian prior is given the empirical covariance of the samples generated by GPSS as its covariance parameter $\Sigma$. Untuned ESS corresponds, as usual, to $\Sigma := I_d$. \label{Fig:hyperplane}}
\end{figure}

\begin{figure}[tb]
	\begin{center}
		\includegraphics[width=\textwidth]{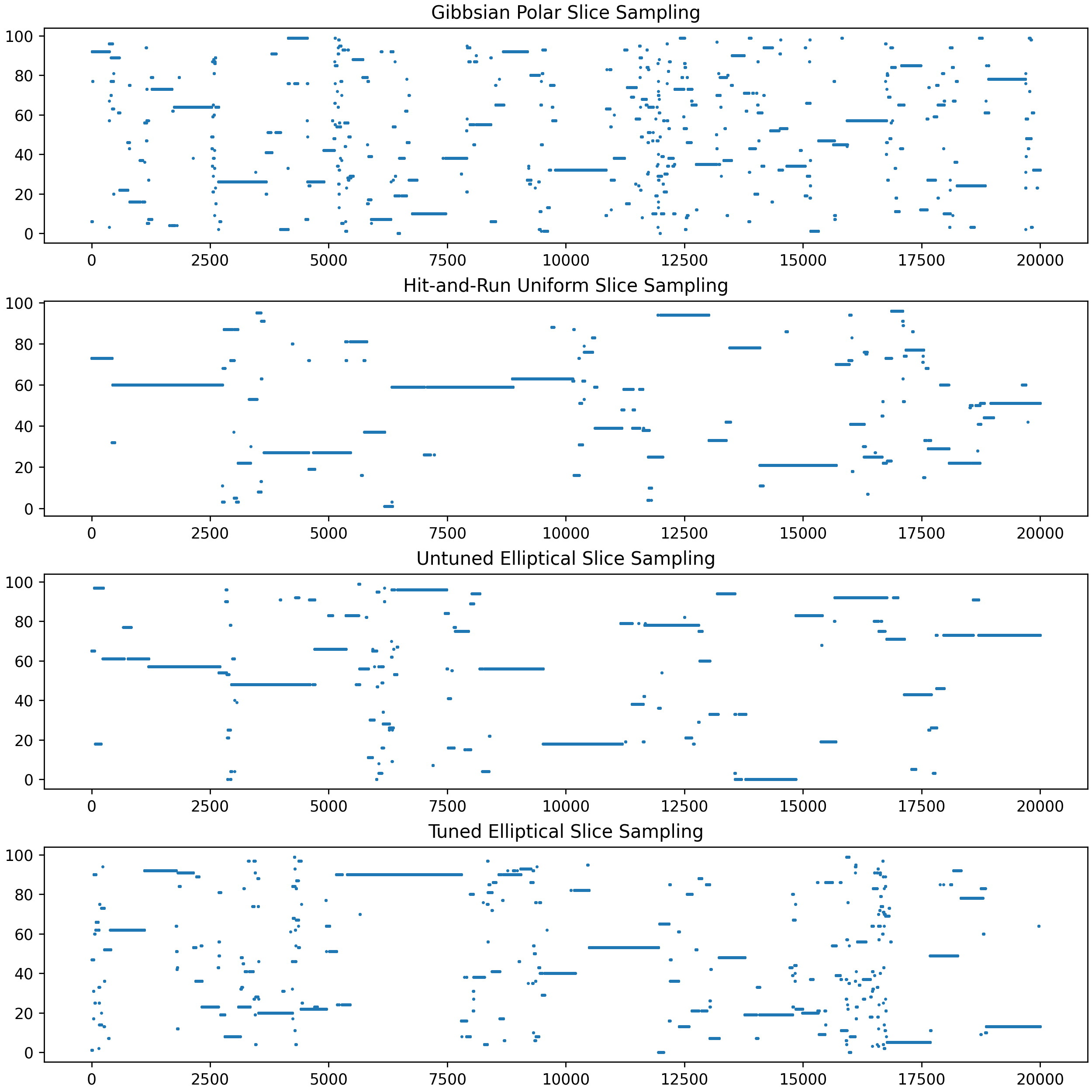}
	\end{center}
	\caption{Progression of currently visited mode pair over the last $N_{\text{window}} = 2 \cdot 10^4$ iterations for the axial modes target density \eqref{Eq:axial_modes_target} in dimension $d=100$. Here the covariance used by untuned ESS was $\Sigma = I_d$ and that used by tuned ESS $\Sigma = (5+d/10)^2 / d \cdot I_d$. \label{Fig:axial_modes_fixed_dim}}
\end{figure}

\begin{figure}[tb]
	\begin{center}
		\includegraphics[width=\textwidth]{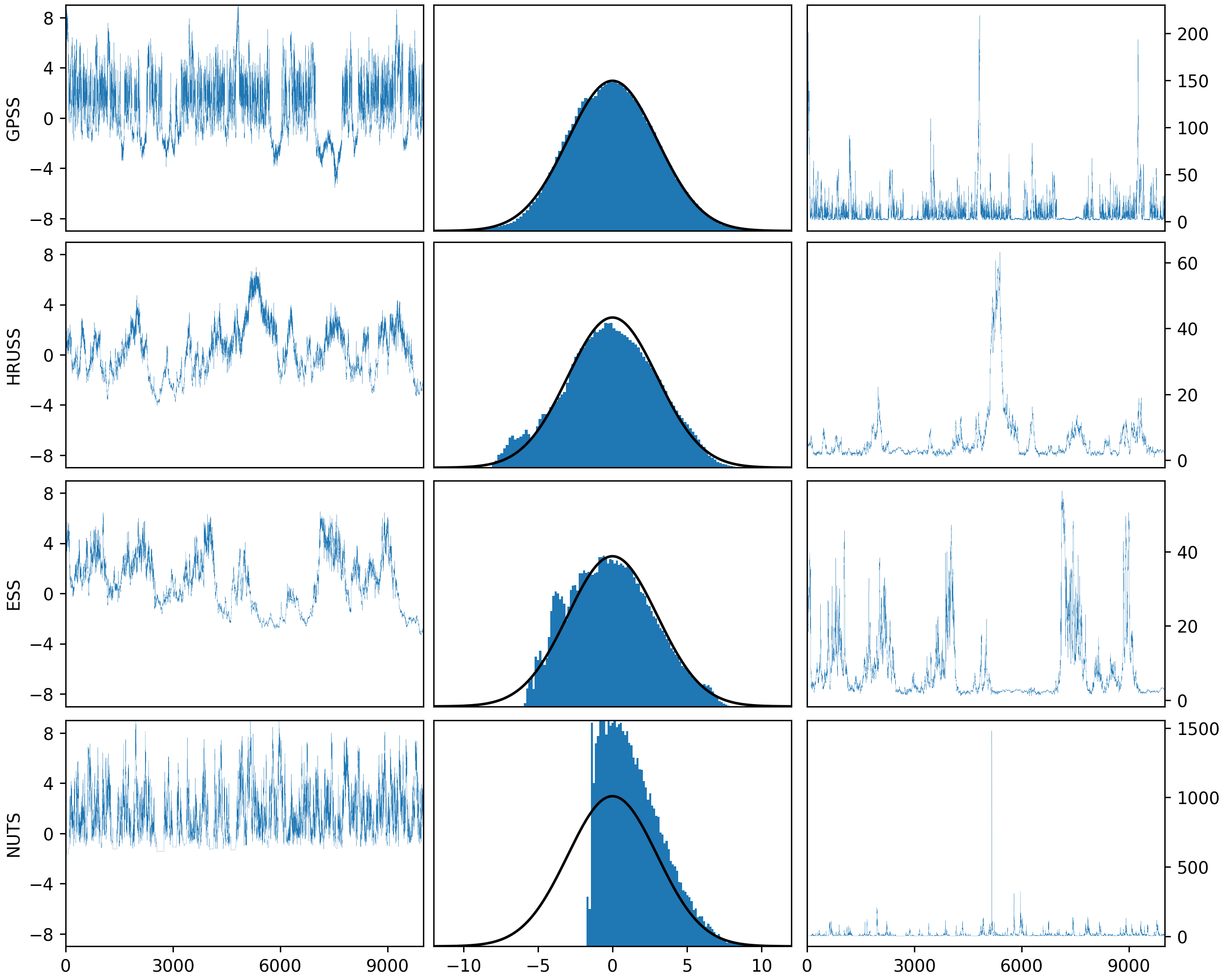}
	\end{center}
	\caption{Marginal traces, marginal histograms and sample radii for Neal's funnel \eqref{Eq:funnel} in dimension $d=10$. The plots in the left column show the progression of the first coordinate component of each sample over the course of each sampler's final $N_{\text{window}} = 10^4$ iterations. The plots in the middle column display histograms of the first coordinate component of all samples each sampler generated within the awarded time budget, with the thick black line marking the target marginal distribution. The plots in the right column show the progression of sample radii (Euclidean norms) over the course of each sampler's final $N_{\text{window}} = 10^4$ iterations. In particular, the quantities displayed in the left and right column of each row are derived from the same $N_{\text{window}}$ samples. \label{Fig:funnel_qualitative}}
\end{figure}
	
\end{document}

%% file: macros.tex
\usepackage{amsmath}
\usepackage{amsfonts}
\usepackage{amssymb}
\usepackage{amsthm}
\usepackage{dsfont} 

\usepackage{array}
\newcolumntype{C}{>{$}c<{$}}

\newcommand{\ra}{\rightarrow}

\newcommand{\N}{\mathbb{N}}

\newcommand{\R}{\mathbb{R}}

\renewcommand{\d}{\text{d}}

\newcommand{\sph}{\mathbb{S}}

\newcommand{\B}{\mathcal{B}}

\newcommand{\Nc}{\mathcal{N}}
\newcommand{\U}{\mathcal{U}}

\DeclareMathOperator*{\argmax}{argmax}

\providecommand\ol[1]{}
\renewcommand\ol[1]{\overline{#1}}

\newcommand{\ind}{\mathds{1}}

\newcommand\ccint[2]{\,[#1,#2]}
\newcommand\ooint[2]{\,(#1,#2)}

\newcommand\coint[2]{\,[#1,#2)}

